\RequirePackage{tikz}
\RequirePackage[noend]{algpseudocode}
\documentclass[sn-mathphys]{sn-jnl}

\usepackage{graphicx}
\usepackage{stmaryrd}
\usepackage{amsmath}
\usepackage{listings}
\usepackage{enumerate}
\usepackage{mathtools} 
\usepackage{algorithm}
\usepackage{hyperref}
\usepackage{xspace}
\usepackage{amsfonts}
\usepackage{subcaption}
\usepackage{pdfpages}
\usepackage[capitalise,nameinlink,noabbrev]{cleveref}
\usepackage{tabularx}
\usepackage{booktabs}

\usetikzlibrary{calc, positioning, decorations, decorations.pathreplacing, decorations.pathmorphing, patterns, arrows, arrows.meta,math}
\usepackage{csquotes}

\usepackage[sort&compress,numbers]{natbib}

\jyear{2021}%

\theoremstyle{thmstyleone}%
\newtheorem{theorem}{Theorem}
\newtheorem{lemma}{Lemma}
\newtheorem{corollary}{Corollary}
\theoremstyle{thmstyletwo}%
\newtheorem{remark}{Remark}%

\theoremstyle{thmstylethree}%
\newtheorem{definition}{Definition}%

\newcommand*{\modelname}{\ensuremath{SCS}\xspace}
\newcommand*{\modelnameconst}{\ensuremath{\modelname^\dynConc}\xspace}
\newcommand*{\modelnameunit}{\ensuremath{\modelname^1}\xspace}
\newcommand*{\modelnameextended}{\ensuremath{\modelname^e}\xspace}

\newcommand*{\jobs}{\ensuremath{\mathcal{J}}\xspace}
\newcommand*{\procS}{\ensuremath{p_s}\xspace}
\newcommand*{\procSF}[1]{\ensuremath{p_s(#1)}\xspace}
\newcommand*{\procC}{\ensuremath{p_c}\xspace}
\newcommand*{\procCF}[1]{\ensuremath{p_c(#1)}\xspace}
\newcommand*{\procCFmc}[2]{\ensuremath{p_{c_#2}(#1)}\xspace}
\newcommand*{\procPiF}[1]{\ensuremath{p^\pi(#1)}\xspace}
\newcommand*{\com}{\ensuremath{c}\xspace}
\newcommand*{\comF}[1]{\ensuremath{c(#1)}\xspace}
\newcommand*{\comFmc}[3]{\ensuremath{c_{#1\triangleright#2}(#3)}\xspace}
\newcommand*{\deadline}{\ensuremath{d}\xspace}
\newcommand*{\budget}{\ensuremath{b}\xspace}
\newcommand*{\completion}[1]{\ensuremath{C(#1)}\xspace}
\newcommand*{\source}{\ensuremath{\mathcal{S}}\xspace}
\newcommand*{\sink}{\ensuremath{\mathcal{T}}\xspace}
\newcommand*{\tasksServer}{\ensuremath{\jobs^s}\xspace}
\newcommand*{\tasksCloud}{\ensuremath{\jobs^c}\xspace}
\newcommand*{\tasksEdges}{\ensuremath{E^*}\xspace}
\newcommand*{\dynProgName}{\textsc{DPfGG}\xspace}
\newcommand*{\scale}{\ensuremath{\varsigma}\xspace}
\newcommand*{\scaled}[1]{\ensuremath{\hat{#1}\xspace}}
\newcommand*{\scaledProcS}[1]{\ensuremath{\hat{p}_s(#1)\xspace}}
\newcommand*{\scaledProcC}[1]{\ensuremath{\hat{p}_c(#1)\xspace}}
\newcommand*{\scaledComF}[1]{\ensuremath{\hat{c}(#1)\xspace}}

\newcommand*{\cost}{\ensuremath{cost}\xspace}
\newcommand*{\costAlg}{\ensuremath{cost_{ALG}}\xspace}
\newcommand*{\costOpt}{\ensuremath{cost_{OPT}}\xspace}
\newcommand*{\makespan}{\ensuremath{mspan}\xspace}
\newcommand*{\makeAlg}{\ensuremath{mspan_{ALG}}\xspace}
\newcommand*{\makeOpt}{\ensuremath{mspan_{OPT}}\xspace}
\newcommand*{\dynTimestamp}{\ensuremath{t}\xspace}
\newcommand*{\dynState}{{\ensuremath{state}}\xspace}
\newcommand*{\dynConc}{\ensuremath{\psi}\xspace}
\newcommand*{\dynFreeServer}{\ensuremath{f_s}\xspace}
\newcommand*{\dynLocE}[1]{\ensuremath{loc_{#1}}\xspace}
\newcommand*{\dynLoc}{\ensuremath{loc}\xspace}
\newcommand*{\dynFreeE}[1]{\ensuremath{f_{#1}}\xspace}

\newcommand*{\clause}[1]{\ensuremath{C^\phi_{#1}}\xspace}

\DeclarePairedDelimiter\mset{\lbrace}{\rbrace}
\DeclarePairedDelimiterX\msett[2]{\lbrace}{\rbrace}{ #1 \,\delimsize| \,\mathopen{} #2 }

\makeatletter
\renewcommand{\p@enumii}{}
\makeatother

\begin{document}

\title[Server Cloud Scheduling]{Server Cloud Scheduling\footnote{This work was partially supported by the German Research Foundation (DFG) within the Collaborative Research Centre “On-The-Fly Computing“ under the project number 160364472 --- SFB 901/3.}}


\author[1]{\fnm{Marten} \sur{Maack} {\tiny ORCID:0000-0001-7918-6642}}\email{martenm@mail.upb.de}

\author[1]{\fnm{Friedhelm} \sur{Meyer auf der Heide}}\email{fmadh@mail.upb.de}

\author*[1]{\fnm{Simon} \sur{Pukrop} {\tiny ORCID:0000-0002-4473-5215}}\email{simonjp@mail.upb.de}

\affil[1]{\orgdiv{Heinz Nixdorf Institute \& Department of Computer Science}, \orgname{Paderborn University}, \orgaddress{\city{Paderborn}, \country{Germany}}}


\abstract{
	Consider a set of jobs connected to a directed acyclic task graph with a fixed source and sink. 
	The edges of this graph model precedence constraints and the jobs have to be scheduled with respect to those. 
	We introduce the Server Cloud Scheduling problem, in which the jobs have to be processed either on a single local machine or on one of infinitely many cloud machines. 
	For each job, processing times both on the server and in the cloud are given. 
	Furthermore, for each edge in the task graph, a communication delay is included in the input and has to be taken into account if one of the two jobs is scheduled on the server and the other in the cloud. 
	The server processes jobs sequentially, whereas the cloud can serve as many as needed in parallel, but induces costs.
	We consider both makespan and cost minimization.
	The main results are an FPTAS for the makespan objective for graphs with a constant source and sink dividing cut and strong hardness for the case with unit processing times and delays.}




\maketitle

\section{Introduction}

Scheduling with precedence constraints with the goal of makespan minimization is widely considered a fundamental problem.
It has already been studied in the 1960s by Graham \cite{Graham66} and receives a lot of research attention up to this day (see e.g. \cite{DBLP:conf/stoc/LeveyR16,DBLP:conf/icalp/Garg18,DBLP:conf/soda/KulkarniLTY20}).
One problem variant that has received particular attention recently, is the variant with communication delays (e.g. \cite{DBLP:conf/soda/KulkarniLTY20,DBLP:conf/focs/DaviesKRTZ20,DBLP:conf/soda/DaviesKRTZ21}).
Another, more contemporary topic concerns scheduling using external resources like, for instance, machines from the cloud and several models in this context have been considered of late (e.g. \cite{DBLP:conf/europar/AbaKP19,DBLP:conf/fsttcs/Saha13,DBLP:journals/jco/MackerMHR18}).
In this paper, we introduce and study a model closely connected to both settings, where jobs with precedence constraints may either be processed on a single server machine or on one of many cloud machines.
Here, communication delays may occur only if the computational setting is changed.
The server and cloud machines may behave heterogeneously, i.e., jobs may have different processing times on the server and in the cloud, and scheduling in the cloud incurs costs proportional to the computational load performed in this context.
Both makespan and cost minimization is considered.
We believe that the present model provides a useful link between scheduling with precedence constraints and communication delays on the one hand and cloud scheduling on the other. 
There is a shorter published conference version \cite{DBLP:conf/waoa/MaackHP21} of this paper; \Cref{sec:extendedchain}, \Cref{sec:generalization} and \Cref{sec:pareto} are new content exclusive to this version.

\subsection{Problem}

We consider a scheduling problem \modelname{} in which a task graph $G=(\jobs, E)$ has to be scheduled on a combination of a local machine (server) and a limitless number of remote machines (cloud).
The task graph is a directed, acyclic graph with exactly one source $\source \in \jobs$ and exactly one sink $\sink \in \jobs$.
Each job $j\in \jobs$ has a processing time on the server $\procSF{j}$ and on the cloud $\procCF{j}$.
We consider $\procSF{\source} = \procSF{\sink} = 0$ and $\procCF{\source} = \procCF{\sink} = \infty$.
For every other job the values of \procS and \procC can be arbitrary in $\mathbb{N}_0$, meaning that the server and the cloud are unrelated machines in our default model.
An edge $e = (i,j)$ denotes precedence, i.e., job $i$ has to be fully processed before job $j$ can start.
Furthermore an edge $e = (i,j)$ has a communication delay of $\comF{i,j} \in \mathbb{N}_0$, which means that after job $i$ finished, $j$ has to wait an additional $\comF{i,j}$ time steps before it can start, if $i$ and $j$ are not both scheduled on the same type of machine (server or cloud).

A schedule $\pi$ is given as a tuple $(\tasksServer, \tasksCloud, C)$.
\tasksServer and \tasksCloud are a proper partition of \jobs: $\tasksServer \cap \tasksCloud = \emptyset$ and $\tasksServer \cup \tasksCloud = \jobs$.
The sets \tasksServer and \tasksCloud denote jobs that are processed on the server or cloud in $\pi$, respectively.
Lastly, $C: \jobs \mapsto \mathbb{N}_0$ maps jobs to their completion time.

We introduce some notation before we formally define the validity of a schedule.
Let \procPiF{j} be equal to \procSF{j} iff $j \in \tasksServer$, and \procCF{j} iff $j \in \tasksServer$.
The value \procPiF{j} denotes the actual processing time of job $j$ in $\pi$.
Let $\tasksEdges := \{(i,j) \in E \mid ( i \in \tasksServer \land j \in \tasksCloud )\lor( i \in \tasksCloud \land j \in \tasksServer )\}$ be the set of edges between jobs on different computational contexts (server or cloud).
Intuitively, for all the edges in \tasksEdges we have to take the communication delays into consideration, for all edges in $E \setminus \tasksEdges$ we only care about the precedence.

We call a schedule $\pi$ valid if and only if the following conditions are met:
\begin{enumerate}
	\item[a)] There is always at most one job processing on the server:\\
	$\forall_{i\in \tasksServer}
	~\forall_{j\in \tasksServer\setminus\{i\}} :
	(\completion{i} \leq \completion{j}-\procPiF{j}) \vee (\completion{i}-\procPiF{i} \geq \completion{j}) $
	\item[b)] Tasks are not started before the preceding tasks have been finished and the required communication is done:\\
	$\forall_{(i,j) \in E \setminus \tasksEdges}: (\completion{i} \leq \completion{j}-\procPiF{j})$\\
	$\forall_{(i,j) \in \tasksEdges}: (\completion{i}+\comF{i,j} \leq \completion{j}-\procPiF{j})$
\end{enumerate}

The makespan ($\makespan$) of a schedule is given by the completion time of the sink \completion{\sink}.
The cost ($\cost$) of a schedule is given by the time it spends processing tasks on the cloud: $\sum_{i \in \tasksCloud} \procPiF{i}$.
Note here, that by requiring  $\procSF{\source} = \procSF{\sink} = 0$ and $\procCF{\source} = \procCF{\sink} = \infty$, we assume every job to start and end on the server.
This is done only for convenience as it defines a clear start and end state for each schedule.

Naturally two different optimization problems arise from the definition. 
First, given a deadline \deadline, find a schedule with lowest cost and $\makespan = \completion{\sink} \leq \deadline$.
Second, given a cost budget \budget, find a schedule with smallest makespan and $\cost = \sum_{i \in \tasksCloud} \procPiF{i} \leq \budget$.
In both instances the \deadline, respectively the \budget, is strict.
The natural decision variant is: given both \deadline and \budget find a schedule that adheres to both, if one exists.

\begin{remark}
	Instances of \modelname might contain schedules with a makespan (and therefore cost) of $0$.
	We can check for those in polynomial time:
	First, remove all edges with communication delay 0, we get a set of connected components $K$.
	Iff $\forall_{k \in K} \left( \forall_{j\in k} ~\procSF{j}=0 \right) \lor \left( \forall_{j\in k} ~\procCF{j}=0 \right)$, then there is a schedule with makespan of $0$.
	For the rest of the paper we will assume that our algorithms check that beforehand and are only interested in schedules with $\makespan >0$.
\end{remark}

\subsection{Results}

We start by establishing (weak) NP-hardness already for the case without communication delays and very simple task graphs.
More precisely, for the case in which the task graph forms one chain starting with the source and ending with the sink and the case in which the graph is fully parallel, i.e., each job $j\in\jobs\setminus\mset{\source,\sink}$ is only preceded by the source and succeeded by the sink.
On the other hand, we establish FPTAS results for both the chain and fully parallel case with arbitrary communication delays and with respect to both objective functions.
Furthermore, we present a $2$-approximation for the case without delays and identical server and cloud machines ($\procC = \procS$) but arbitrary task graph and the makespan objective and show that the respective algorithm can also be used to solve the problem optimally with respect to both objectives in the case of unit processing times.
These results are all relatively simple and are discussed in \cref{sec:prelim}.
In \Cref{sec:extendedchain} we generalize the previous two task graph models (chain and fully parallel) into one, called extended chain graphs.
We present a $(2+\varepsilon)$-approximation for the budget restrained makespan minimization for this class of task graphs.
Furthermore, we discuss some small assumptions on the problem instance, which allow us to achieve FPTAS results instead.
We end the section by giving a reduction from the strongly NP-hard $1\mid r_j\mid \sum w_j U_j$ problem \cite{lenstra1977complexity}.
In \cref{sec:constant_width_FPTAS} we aim to generalize the previous FPTAS results regarding the makespan as much as possible.
We are able to show that an FPTAS can be achieved as long as the \emph{maximum cardinality source and sink dividing cut} \dynConc is constant.
Intuitively, this parameter upper bounds the number of edges that have to be considered together in a dynamic program and in many relevant problem variants it can be bounded or replaced by the longest anti-chain length.
We provide a formal definition in \cref{sec:constant_width_FPTAS}.
Next, we turn our attention to strong NP-hardness results in \cref{sec:strong_hardness}. 
We are able to show, that a classical reduction due to Lenstra and Rinnooy Kan \cite{DBLP:journals/ior/LenstraK78} can be adapted to prove NP-hardness already for the variant of \modelname without communication delays and processing times equal to one or two.
Now, in the case of unit processing times without communication delays this can be trivially solved in polynomial time, and hence we are interested in the case with unit processing times and communication delays.
We design an intricate reduction to show that this very basic case is NP-hard as well.
Note that in this setting the server and cloud machines are implicitly identical.
Furthermore, we are able to show that a slight variation of this reduction implies that no constant approximation with respect to the cost objective can be achieved regarding the general problem. 
In \cref{sec:unit_size_unit_delay}, we consider approximation algorithms for the case with unit processing times and delays.
We show that a relatively simple approach yields a  $\frac{1+\varepsilon}{2\varepsilon}$-approximation for $\varepsilon\in(0,1]$ regarding the cost objective if we allow a makespan of $(1+\varepsilon)\deadline$.
In \cref{sec:generalization}, we establish some natural generalizations on the model and sketch how those can be solved by slight adaptations of our algorithms for extended chain and constant $\phi$ graphs.
Lastly, in \cref{sec:pareto} we show how to give an $\alpha$-approximation, for any chosen $\alpha > 0$, on the pareto front of a problem with a task graph with constant $\phi$, when we look at the problem as a multi objective optimization problem.
This means, that for any point in the actual pareto front, we give a nearby feasible point that is only worse by a factor of $1+\alpha$ in both dimensions. 
In \Cref{table:results} we give an overview over the important results.

\begin{table}
	\caption{An overview of the results of this paper.}
	\centering
	\begin{tabularx}{\columnwidth}{Xp{5.3cm}}
		\toprule
		Algorithmic Results &  \\ \midrule
		fully parallel or chain task graph & FPTAS w.r.t. cost and makespan\\
		extended chain task graph & $(2+\varepsilon)$-approximation w.r.t. makespan\\
		extended chain + additional assumptions &FPTAS w.r.t. makespan\\
		extended chain task graph + generalizations & $(4+\varepsilon)$-approximation w.r.t. makespan\\
		task graph with constant \dynConc & FPTAS w.r.t. makespan \\
		task graph with constant \dynConc & $\alpha$-approximation of Pareto front, for any $\alpha>0$ \\
		task graph with constant \dynConc + generalizations & FPTAS w.r.t. makespan \\
		$\com = 0$, $\procC = \procS$ (no delays, identical machines) & $2$-approximation w.r.t. makespan \\
		$\com = 0$, $\procC = \procS=1$ & polynomial w.r.t. makespan and cost \\
		$\com =\procC = \procS = 1$ (unit delays, unit sizes) & $\frac{1+\varepsilon}{2\varepsilon}$-approximation w.r.t. cost with makespan at most $(1+\varepsilon)\deadline$ \\ \midrule
		Hardness Results & \\\midrule
		fully parallel or chain task graph, $\com = 0$ &  (weakly) NP-hard \\
		extended chain task graph &  (strongly) NP-hard \\
		$\forall j\in\jobs: \com(j) = 0, \procC(j),\procS(j)\in\mset{1,2}$ &  (strongly) NP-hard \\
		$\com =\procC = \procS = 1$ (unit delays, unit sizes) & (strongly) NP-hard \\
		general problem & no constant approximation w.r.t. cost \\ \bottomrule
	\end{tabularx}
	\label{table:results}
\end{table}

\subsection{Related Work}

Probably the closest related model to the one considered in this paper was studied by Aba et al. \cite{DBLP:conf/europar/AbaKP19}.
In this paper the input is very similar, however, in both computational settings an unbounded number of machines may be used and the goal is makespan minimization.
The authors show NP-hardness on the one hand, and identify cases that can be solved in polynomial time on the other.
In the conclusion of this paper a model very similar to the one studied in this work is mentioned as an interesting research direction.
For a detailed discussion of related models, we refer to the preprint version of the above work~\cite{DBLP:conf/europar/AbaKP19}.

The present model is closely related to the classical problem of makespan minimization on parallel machines with precedence constraints, 
where a set of jobs with processing times, a precedence relation on the jobs (or a task graph), and a set of $m$ machines are given.
The goal is to assign the jobs to starting times and machines such that the precedence constraints are met and the last job finishes as soon as possible.
In the 1960's, Graham \cite{Graham66} introduced the list scheduling heuristic for this problem and proved it to be a $(2-\frac{1}{m})$-approximation.
Interestingly, to date, this is essentially the best result for the general problem.
On the other hand, Lenstra and Rinnooy Kan \cite{DBLP:journals/ior/LenstraK78} showed that no better than $\frac{4}{3}$-approximation can be achieved for the problem with unit processing times, unless P=NP.
In more recent days, there has been a series of exciting new results for this problem starting with a paper by Svensson \cite{DBLP:journals/siamcomp/Svensson11} who showed that no better than $2$-approximation can be hoped for assuming a variant of the unique games conjecture.
Furthermore, Levey and Rothvoss \cite{DBLP:conf/stoc/LeveyR16} presented an approximation scheme with nearly quasi-polynomial running time for the variant with unit processing times and a constant number of machines, and Garg \cite{DBLP:conf/icalp/Garg18} improved the running time to quasi-polynomial shortly thereafter. 
These results utilized so called LP-hierarchies to strengthen linear programming relaxations of the problems. 
This basic approach has been further explored in a series of subsequent works (e.g. \cite{DBLP:conf/soda/KulkarniLTY20,DBLP:conf/focs/DaviesKRTZ20,DBLP:conf/soda/DaviesKRTZ21}), which in particular also investigate the problem variant where a communication delay is incurred for pairs of precedence-constrained jobs running on different machines.
The latter problem variant is closely related to our setting as well. 

Lastly, there is at least a conceptual relationship to problems where jobs are to be executed in the cloud. 
For example, a problem was considered by Saha \cite{DBLP:conf/fsttcs/Saha13} in which cloud machines have to be rented in fixed time blocks in order to schedule a set of jobs with release dates and deadlines minimizing the costs which are proportional to the rented time blocks.
Another example is a work by M\"acker et al. \cite{DBLP:journals/jco/MackerMHR18} in which machines of different types can be rented from the cloud and machine dependent setup times have to be payed before they can be used. 
Jobs arrive in an online fashion and the goal is again cost minimization.
Both papers reference further work in this context.

\section{Preliminary Results - Chains and Fully Parallel}\label{sec:prelim}

In this section we collect some results that can be considered low hanging fruits and give a first overview concerning the complexity and approximability of our problem. 
In particular, we show weak NP-hardness already for cases with very simple task graphs and without communication delays. 
Furthermore, we discuss complementing FPTAS results and a $2$-approximation for the case with identical cloud and server machines and without communication delays. 

\subsection{Hardness}
\label{sec:prelimHardness}
We show that \modelname is NP-hard even for two very simple types of taskgraphs and in a case where every communication time is $0$.
For both of these reductions we use the decision variant of the problem: given both a deadline \deadline and a budget \budget, find a schedule that satisfies both.
Naturally this will show the hardness of both the cost minimization as well as the makespan minimization problem.
We start by reducing the decision version of knapsack to \modelname with a chain graph as its task graph.
The knapsack problem is given as a capacity $C$, a value threshold $V$ and a set of items $\{1,\dots,n\}$ with weights $w_i$ and values $v_i$.

The question is, if there exist is a subset of items $S$ such that $\sum_{i\in S}w_i \leq C$ and $\sum_{i\in S}v_i \geq V$.
We create the respective \modelname problem as follows.
For every item $i \in \{1,\dots,n\}$ create a task with $\procSF{i}=w_i+v_i$ and $\procCF{i}=v_i$. Consider a task graph with those tasks as a chain (in an arbitrary order) and each resulting edge $(i,j)$ has $\comF{i,j}=0$.
We set the deadline to $\deadline = \sum_{1\leq i \leq n}v_i+C$ and the budget to $\budget = \sum_{1\leq i \leq n}v_i-V$.
It is left to show, that there is a solution to the knapsack problem if and only if there is a schedule to our transformed problem.
Basically we show that there is a one to one relation between our schedules and knapsack solutions.
Assume there is some feasible solution (subset of items $S$) for the knapsack problem with value $V'$.
For each $i\in S$ we put the respective task in \tasksServer and the rest in \tasksCloud.
Since the task graph is a chain we can compute a minimal makespan from this partition: $\sum_{1\leq i \leq n}v_i + \sum_{i \in S} w_i$ which is smaller or equal to $\deadline$ if and only if $\sum_{i \in S} w_i \leq C$.
The cost for the schedule is equal to $\sum_{1\leq i \leq n}v_i - V'$.
Therefore, the cost for the schedule is smaller or equal to $\budget$ exactly when $V' \geq V$.
It is easy to see that we can construct a knapsack solution from a schedule in a similar vein, therefore we conclude:
\begin{theorem}
	\label{the:chainNPhard}
	The \modelname problem is weakly NP-hard for chain graphs and without communication delays.
\end{theorem}

Secondly we look at problems with fully parallel task graphs, which means that every job $j$ besides \source and \sink has exactly two edges: $(\source,j)$ and $(j,\sink)$.
Here we do a simple partition reduction.
Given a set $S$ of natural numbers, the question is, if there is a partition into sets $S_1$ and $S_2$ such that $\sum_{i \in S_1} i = \sum_{i \in S_2} i$?
For every element $i$ in $S$ we create a task with $\procSF{j}=\procCF{j}=i$, set $\deadline = \budget = \frac{1}{2}\sum_{i \in S_1} i$.
We arrange the tasks into a fully parallel task graph where each edge $(i,j)$ has $\comF{i,j}=0$.
Imagine a solution $S_1$, $S_2$ for the partition problem.
We schedule every task related to an integer in $S_1$ on the server and every other task on the cloud.
Since everything is fully parallel and there are no communication delays we can conclude a makespan of $\max \{\sum_{i \in S_1} i, \max_{i \in S_2} i\}$ and costs of $\sum_{i \in S_2}$.
This is a correct solution for the scheduling problem if and only if $\sum_{i \in S_1} i = \sum_{i \in S_2} i$.
Again it is easy to see that an equivalent argument can be made for the other direction.
\begin{theorem}
	\label{the:fullyparallelNPhard}
	The \modelname problem is weakly NP-hard for fully parallel graphs and without communication delays.
\end{theorem}


\subsection{Algorithms}

In the following, we present complementing FPTAS results for the variants of \modelname with fully parallel and chain task graphs. 
Furthermore, in both of the above reductions we did have no communication delays and in one of them the jobs had the same processing time on the server and the cloud.
Hence, we take a closer look at this case as well and present a simple $2$-approximation even for arbitrary task graphs and with respect to the makespan objective.

\subsubsection{Fully Parallel Case}

We show that the variant of \modelname with fully parallel task graph can be dealt with using straight-forward applications of well-known results and techniques.
In particular, we can design two simple dynamic programs for the search version of the problem that consider for each job the two possibilities of scheduling them on the cloud or on the server and compute for each possible budget or deadline the lowest makespan or cost, respectively, that can be achieved with the jobs considered so far.
These dynamic programs can then be combined with suitable rounding procedures that reduce the number of considered states and search procedures for approximate values for the optimal cost or makespan, respectively, yielding:
\begin{theorem}
	There is an FPTAS for \modelname with fully parallel task graph with respect to both the cost and the makespan objective.
\end{theorem}
\begin{proof}
	We start by designing the dynamic programs for the search version of the problem with budget $\budget$ and deadline $\deadline$.
	Without loss of generality, we assume $\jobs =\mset{0,1,\dots,n,n+1}$ with $\source = 0$, $\sink = n+1$ and set $\comF{j} = \comF{\source,j} + \comF{j,\sink}$.
	
	For each deadline $d'\in\mset{0,1,\dots,\deadline}$ and $j\in \jobs$, we want to compute the smallest cost $C[j,d']$ of all the schedules of the jobs $0,1,\dots,j$ adhering to the deadline $d'$ on the server ($j = 0$ denotes the trivial case that no job after the source has been scheduled).
	We initialize $C[0,d']=0$ for each $d'$.
	For all other jobs $j$ we consider the two possibilities of scheduling it on the cloud or server.
	In particular, let $C_1[j,d'] = C[j-1,d']+\procC(j)$ if $\procC(j) + \comF{j} \leq \deadline$ and $C_1[j,d'] = \infty$ otherwise, and, furthermore, $C_2[j,d'] = C[j-1,d'-\procS(j)]$ if $\procS(j)\leq d'$ and $C_2[j,d'] = \infty$ otherwise. 
	Then, we may set $C[j,d'] = \min\mset{C_1(j,d'),C_2(j,d')}$. 
	Now, if $C[n+1,\deadline] > \budget$, we know that there is no feasible solution for the search version, and otherwise we can use backtracking starting from $C[n+1,\deadline]$ to find one.
	The time and space complexity is polynomial in $\deadline$ and $n$.
	
	In the second dynamic program, we compute the smallest makespan $M[j,b']$ of all the schedules of the jobs $0,1,\dots,j$ adhering to the budget $b'$, for each budget $b'\in\mset{0,1,\dots,\budget}$ and $j\in \jobs$.
	Again, we set  $M[0,b']=0$ for each $b'$ and consider the two possibilities of scheduling job $j$ on the cloud or server.
	To that end, let $M_1[j,b'] = \max\mset{M[j-1,b'-\procC(j)], \procCF{j} + \comF{j}}$ if $\procC(j) + \comF{j} \leq d$ and $b'-\procC(j) \geq 0$. 
	Otherwise, set $M_1[j,b'] = \infty$, furthermore, $M_2[j,b'] = M[j-1,b'] + \procS(j)$.
	Then, we may set $M[j,b'] = \min\mset{M_1(j,b'),M_2(j,b')}$. 
	Again, if $M[n+1,\budget] > \deadline$, we know that there is no feasible solution for the search version, and otherwise we can use backtracking starting from $M[n+1,\budget]$ to find one.
	The time and space complexity is polynomial in $\budget$ and $n$.
	
	For both programs, we can use rounding and scaling approaches to trade the complexity dependence in $\deadline$ or $\budget$ with a dependence in $poly(n,\frac{1}{\varepsilon})$ incurring a loss of a factor $(1+\mathcal{O}(\varepsilon))$ in the makespan or cost, respectively, if a solution is found.
	This can then be combined with a suitable search procedure for approximate values of the optimal makespan or cost.
	For details, we refer to \cref{sec:constant_width_FPTAS}, where such techniques are used and described in more detail.
	In addition to the techniques mentioned there, the possibility of a cost zero solution has to be considered which can easily be done in this case.
\end{proof}

\subsubsection{Chain Graph Case}

We present FPTAS results for the variant of \modelname with chain task graph.
The basic approach is very similar to the fully parallel case. 
\begin{theorem}
	\label{the:fptasChain}
	There is an FPTAS for \modelname with chain task graph with respect to both the cost and the makespan objective.
\end{theorem}
\begin{proof}
	We again start by designing dynamic programs for the search version of the problem with budget $\budget$ and deadline $\deadline$.
	Without loss of generality, we assume $\jobs =\mset{0,1,\dots,n+1}$ with $\source = 0$, $\sink = n+1$, and $j\in \mset{0,1,\dots,n+1}$ being the $j$-th job in the chain.
	
	For each deadline $d'\in\mset{0,1,\dots,\deadline}$, job $j\in\mset{0,1,\dots,n+1}$, and location $\dynLoc\in\mset{s,c}$ (referring to the server and cloud) we want to compute the smallest cost $C[d',j,\dynLoc]$ of all the schedules of the jobs $1,\dots,j$ adhering to the deadline $d'$ and with the job $j$ being scheduled on $\dynLoc$.
	To that end, we set $C[d',0,s] = 0$, $C[d',0,c] = \infty$, and with slight abuse of notation use the convention $C[z,j,\dynLoc] = \infty$ for $z<0$.
	Further values can be computed via the following recurrence relations:
	\begin{align*}
	C[d',j,s] &= \min\mset{C[d' - \procS(j) - \comF{j-1,j},c],  C[d' - \procS(j),s]}\\
	C[d',j,c] &= \min\mset{C[d' - \procC(j),c] + \procC(j),  C[d' - \procC(j) - \comF{j-1,j},s] + \procC(j)}
	\end{align*} 
	If $C[\deadline,n+1,s] > \budget$, we know that there is no feasible solution for the search version, and otherwise we can use backtracking starting from $C[\deadline,n+1,s]$ to find one.
	The time and space complexity is polynomial in $\deadline$ and $n$.
	
	In the second dynamic program, we compute the smallest makespan $M[j,b',\dynLoc]$ of all the schedules of the jobs $0,\dots,j$ adhering to the budget $b'$ and with job $j$ placed on location $\dynLoc$, for each $b'\in\mset{0,1,\dots,\budget}$, $j\in\mset{0,1,\dots,n+1}$ and $\dynLoc\in\mset{s,c}$.
	We set $M[b',0,s] = 0$, $M[b',0,c] = \infty$, use the convention $M[z,j,\dynLoc] = \infty$ for $z<0$, and the recurrence relations:
	\begin{align*}
	M[b',j,s] &= \min\mset{M[b',c] + \procS(j) + \comF{j-1,j},  M[b',s] +  \procS(j)}\\
	M[b',j,c] &= \min\mset{M[b' - \procC(j),c] + \procC(j),  M[b' - \procC(j),s] + \procC(j) + \comF{j-1,j}}
	\end{align*} 
	If $M[\budget, n+1 ,s] > \deadline$, we know that there is no feasible solution for the search version, and otherwise we can use backtracking starting from $M[\budget, n+1 ,s]$ to find one.
	The time and space complexity is polynomial in $\budget$ and $n$.
	
	Like in the fully parallel case, we can use rounding and scaling approaches to trade the complexity dependence in $\deadline$ or $\budget$ with a dependence in $poly(n,\frac{1}{\varepsilon})$ incurring a loss of a factor $(1+\mathcal{O}(\varepsilon))$ in the makespan or cost, respectively, if a solution is found.
	This can then be combined with a suitable search procedure for approximate values of the optimal makespan or cost.
	For details, we refer to \cref{sec:constant_width_FPTAS}, where such techniques are used and described in more detail.
	In addition to the techniques mentioned there, the possibility of a cost zero solution has to be considered which can easily be done in this case as well.
	
\end{proof}

\section{The Extended Chain Model} \label{sec:extendedchain}
As a first step towards more general models we introduce the extended chain model.
The main idea here is to find a unifying generalization for the chain and fully parallel case.
Informally one can imagine an extended chain as a chain graph where any number of edges were replaced with fully parallel graphs.
After giving a formal definition of these graphs we introduce a $(2+\varepsilon)$-approximation for the budget restrained makespan minimization.
That algorithm uses reductions to single machine weighted number of tardy jobs scheduling to solve some intermediate parts via known procedures.
Therefore, we briefly discuss this problem here before actually giving our algorithm.
We finish the constructive side by exploring some assumptions on problem instances that allow us to achieve FPTAS results with our approach.
Lastly, we give a reduction to show that this problem is strongly NP-hard.

\subsection{Single Machine Weighted Number of Tardy Jobs}
\newcommand{\algknap}{\texttt{Knapsack}\xspace}
\newcommand{\algtardy}{\texttt{wTardyJobs}\xspace}
\newcommand{\wntj}{\textsc{WNTJ}\xspace}
As mentioned before this section reduces some intermediate steps in the algorithm to the single machine weighted tardiness problems, for which we will reuse an already established algorithm.

The single machine weighted number of tardy jobs (\wntj) problem, or $1\mid ~\mid \sum w_j U_j$ in three field notation \cite{graham1979optimization}, can be defined as follows:
On a single machine, where only one job at a time can be processed, are $n$ jobs to be scheduled. 
Each job has an integer processing time $p_j$, weight $w_j$ and due date $d_j$. 
A job is called 'late' if it is scheduled completion time $C_j > d_j $ and 'early' if $C_j \leq d_j$. The goal is to find a schedule which minimizes the sum over the weights of the tardy (late) jobs.
Pseudo polynomial dynamic programs with runtime in $\mathcal{O}(n\min\{\sum_{j}p_j,\max_j d_j\})$ and $\mathcal{O}(n\min\{\sum_{j}p_j,\sum_{j}w_j,\max_j d_j\})$, respectively, were given by Lawler and Moore \cite{lawler1969functional} and later Sahni \cite{DBLP:journals/jacm/Sahni76}.
Denote the former by \algtardy.
For a more comprehensive survey on this (and related) problems, we refer to \cite{adamu2014survey}.

\subsection{Model}
We give a constructive description of extended chain graphs.
Let $G=(\jobs,E)$ with $\source\in \jobs$ and $\sink \in \jobs$ be a chain graph.
For any number of edges $e=(j-1,j) \in E$ we may remove the edge $e$ and introduce a set of jobs $\jobs_j$ and for every $j' \in \jobs_j$ two edges, namely $(j-1,j')$ and $(j',j)$.
The resulting graph $G'=(\jobs',E')$ is an extended chain graph.
We denote by $N$ the total number of jobs (nodes) in the graph.
Denote the \modelname problem on extended chains by \modelnameextended.
For an example we refer to \Cref{fig:extendedchain}.
Note here, that the introduced subgraphs are fully parallel graphs as described earlier and consequently fully parallel graphs, as well as chain graphs, are a subset of extended chain graphs.
This also directly infers that \modelnameextended is at least weakly NP-hard as shown in \Cref{the:chainNPhard} and \Cref{the:fullyparallelNPhard}.

\begin{figure}
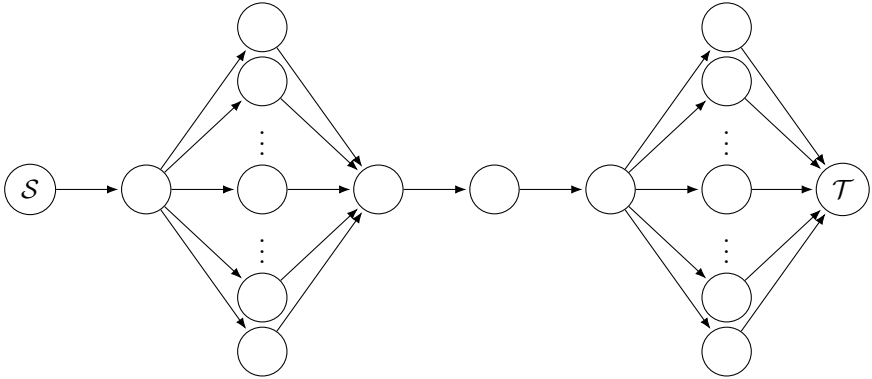

	\centering
	\tikz[>={Latex[length=1.5mm]}, shorten <= 0pt, shorten >= 1pt, scale=0.9]{
		\pgfmathsetmacro{\w}{1.7}
		\pgfmathsetmacro{\h}{0.8}
		\pgfmathsetmacro{\nw}{0.65}
		

		\node[draw, circle] (chain0) at (0,0) {$\source$};
		
		\foreach \x in {1,2,3,4,5,6}
		{
			\node[draw, circle, inner sep=0, minimum size = \nw cm] (chain\x) at (\x*\w,0) {};
		}
		\node[draw, circle] (chain7) at (7*\w,0) {$\sink$};
		
		\foreach \x/\y in {0/1,1/2,2/3,3/4,4/5,5/6,6/7}
		{
			\draw[->] (chain\x) -- (chain\y);
		}
		
		\node[draw, circle, inner sep=0, minimum size = \nw cm] (parallel11) at (2*\w,3*\h) {};
		\node[draw, circle, inner sep=0, minimum size = \nw cm] (parallel12) at (2*\w,2*\h) {};
		\node at (2*\w,\h) {$\vdots$};
		\node at (2*\w,-\h) {$\vdots$};
		\node[draw, circle, inner sep=0, minimum size = \nw cm] (parallel13) at (2*\w,-2*\h) {};
		\node[draw, circle, inner sep=0, minimum size = \nw cm] (parallel14) at (2*\w,-3*\h) {};

		\node[draw, circle, inner sep=0, minimum size = \nw cm] (parallel21) at (6*\w,3*\h) {};
		\node[draw, circle, inner sep=0, minimum size = \nw cm] (parallel22) at (6*\w,2*\h) {};
		\node at (6*\w,\h) {$\vdots$};
		\node at (6*\w,-\h) {$\vdots$};
		\node[draw, circle, inner sep=0, minimum size = \nw cm] (parallel23) at (6*\w,-2*\h) {};
		\node[draw, circle, inner sep=0, minimum size = \nw cm] (parallel24) at (6*\w,-3*\h) {};
		
		\foreach \x in {1,2,3,4}
		{
			\draw[->] (chain1) -- (parallel1\x);
			\draw[->] (parallel1\x) -- (chain3);
			\draw[->] (chain5) -- (parallel2\x);
			\draw[->] (parallel2\x) -- (chain7);
	}}
	\caption{An example extended chain with two parallel parts.}
	\label{fig:extendedchain}
\end{figure}

\subsection{A $(2+\varepsilon)$-approximation for Makespan Minimization on the Extended Chain}

\begin{theorem} \label{the:exChainApprox}
	There is a $(2+\varepsilon)$-approximation algorithm for the budget restrained makespan minimization problem on extended chains.
\end{theorem}

\newcommand{\makeEstimate}{\ensuremath{T}\xspace}
We design a pseudo polynomial algorithm, that given a feasible makespan estimate \makeEstimate ($T \geq \makeOpt$) calculates a schedule with makespan at most $\min \{2\makeEstimate, 2\makeOpt \}$.
Otherwise ($T < \makeOpt$) the algorithm calculates a schedule with makespan at most $\min \{2\makeEstimate, 2\makeOpt \}$ or no schedule at all.
We can use a binary search to find $\makeEstimate \approx OPT$, beginning with the trivial upperbound $\makeEstimate = \sum_{j\in\jobs'}\procSF{j} \geq \makeOpt$

We first introduce notation that follows the constructive description of extended chains above.
We assume $\jobs =\mset{0,1,\dots,n+1}$ with $\source = 0$, $\sink = n+1$, and $j\in \{1, \dots, n \}$ being the $j$-th job in the original chain.
If there is a parallel subgraph between some jobs $j-1$ and $j$ we denote the jobs in it by $\jobs_{j} = \mset{0^j,1^j,\dots,m^j}$.

We reuse the state description from \Cref{the:fptasChain}, but this time we iteratively create all reachable states by going over the jobs $\mset{0,1,\dots,n+1}$.
A state is a combination of timestamp $ t \in \mset{0,1,\dots,\makeEstimate}$, job $j\in\mset{0,1,\dots,n+1}$, and location $\dynLoc\in\mset{s,c}$ (referring to server and cloud respectively).
The value of a state is the smallest cost of all the schedules of the jobs $0, 1, \dots, j$ finishing processing during or before timestamp $t$, with $j$ being scheduled on $\dynLoc$, denoted by $[t,j,\dynLoc] = \cost$.
Note, that we have not mentioned the parallel subgraphs in the description above.
We start with the trivial start state $[0,0(=\source),s] = 0$

Let $\textsc{StateList}^{j-1}$ be the list of states for some job of the chain $j-1$.
We create $\textsc{StateList}^{j}$ in the following way:
First we create a set of state extensions $\textsc{Extensions}^{j}$, each of form $[\Delta t,\dynLocE{j-1}\rightarrow \dynLocE{j}] = \cost$.
Then we form every (fitting) combination of a state from $\textsc{StateList}^{j-1}$ with an extension from $\textsc{Extensions}^{j}$, which forms $\textsc{StateList}^{j}$.
Lastly we cull all dominated states from $\textsc{StateList}^{j}$ and continue with $j+1$.

Calculate $\textsc{Extensions}^{j}$:
\begin{enumerate}
	\item If there is no parallel subgraph between $j-1$ and $j$ we can simply enumerate all state extensions:
	\begin{enumerate}
		\item $j-1$ on server, $j$ on server: $[\procSF{j},~s\rightarrow s] = 0$
		\item $j-1$ on server, $j$ on cloud: $[\procCF{j} + \comF{j-1,j},~s\rightarrow c] = \procCF{j}$
		\item $j-1$ on cloud, $j$ on server: $[\procSF{j} + \comF{j-1,j},~c\rightarrow s] = 0$
		\item $j-1$ on cloud, $j$ on cloud: $[\procCF{j},~c\rightarrow c] = \procCF{j}$
	\end{enumerate}
	\item Otherwise, there is a parallel subgraph between $j-1$ and $j$ with jobs $\jobs_{j} = \mset{0^j,1^j,\dots,m^j}$.
	\begin{enumerate}
		\item \label{alg:case:s-s} $j-1$ on server, $j$ on server: \\
		Set $\Delta^{max} = \min\{\sum_{j' \in \jobs_{j}} \procSF{j'}, \makeEstimate\} $, for every $\Delta^i$ in $\{0,\dots, \Delta^{max} \}$, do the following:
		Set $\jobs^s=\emptyset$ and $\jobs^c=\emptyset$. 
		For every $j' \in \jobs_{j}$ check:
		\begin{itemize}
			\item $\procSF{j'} > \Delta^i$ and $ \comF{j-1,j'} + \procCF{j'} + \comF{j',j} > \Delta^i$:\\
			break and go to next $\Delta^i$ (state extension $[\Delta^i,~s\rightarrow s]$ not feasible)
			\item $\procSF{j'} > \Delta^i$ and $ \comF{j-1,j'} + \procCF{j'} + \comF{j',j} \leq \Delta^i$:\\
			add $j'$ to $\jobs^c$ ($j'$ has to be put on the cloud)
			\item $\procSF{j'} \leq \Delta^i$ and $ \comF{j-1,j'} + \procCF{j'} + \comF{j',j} > \Delta^i$:\\
			add $j'$ to $\jobs^s$ ($j'$ has to be put on the server)
		\end{itemize}
		If $\sum_{j' \in \jobs^s} \procSF{j'} > \Delta^i$ break and go to next $\Delta^i$.
		Create a \wntj instance as follows:
		For every job $j' \in \jobs_{j} \setminus (\jobs^s \cup \jobs^c)$ create a job $j''$ with processing time $p_{j'} = \procSF{j'}$, deadline $d_{j''} = \Delta^i - \sum_{j' \in \jobs^s} \procSF{j'}$ and weight $w_{j''} = \procCF{j'}$.
		Solve this problem with \algtardy, let $V$ be the cost of the solution.
		Add $[\Delta^i,~s\rightarrow s] = \sum_{j' \in \jobs^c} \procCF{j'} + V$ to $\textsc{Extensions}^{j}$.
		(Remark: This could also be solved as a knapsack problem, but we need \wntj later either way.)
		
		\item \label{alg:case:s-c} $j-1$ on server, $j$ on cloud: \\
		Set $\Delta^{max} = \min\{\sum_{j' \in \jobs_{j}} \procSF{j'} + \max_{j' \in \jobs_{j}} \comF{j',j}, \makeEstimate \} $, for every $\Delta^i$ in $\{0,\dots, \Delta^{max} \}$, do the following:
		Set $\jobs^s=\emptyset$ and $\jobs^c=\emptyset$. 
		For every $j' \in \jobs_{j}$ check:
		\begin{itemize}
			\item $\procSF{j'} + \comF{j',j} > \Delta^i$ and $ \comF{j-1,j'} + \procCF{j'} > \Delta^i$:\\
			break and go to next $\Delta^i$ (state extension $[\Delta^i,~s\rightarrow c]$ not feasible)
			\item $\procSF{j'} + \comF{j',j} > \Delta^i$ and $ \comF{j-1,j'} + \procCF{j'} \leq \Delta^i$:\\
			add $j'$ to $\jobs^c$ ($j'$ has to be put on the cloud)
			\item $\procSF{j'} + \comF{j',j} \leq \Delta^i$ and $ \comF{j-1,j'} + \procCF{j'} > \Delta^i$:\\
			add $j'$ to $\jobs^s$ ($j'$ has to be put on the server)
		\end{itemize} 
		Create a \wntj instance as follows:
		For every job $j' \in \jobs_{j} \setminus \jobs^c$ create a job $j''$ with processing time $p(j'') = \procSF{j'}$, deadline $d_{j''} = \Delta^i - \comF{j',j}$ and weight $w_{j''} = \procCF{j'}$ if $j' \notin \jobs^s$, $w_{j''} = \infty$ otherwise.
		Solve this problem with \algtardy, let $V$ be the cost of the solution, if $V = \infty$ break.
		Otherwise, add $[\Delta^i,~s\rightarrow c] = \sum_{j' \in \jobs^c} \procCF{j'} + V$ to $\textsc{Extensions}^{j}$.
		
		\item \label{alg:case:c-s} $j-1$ on cloud, $j$ on server: \\
		This works analogously to the previous case.
		Simply replace each instance of $\comF{j',j}$ by $\comF{j-1,j'}$ and vice versa.
		Add the resulting extensions to $\textsc{Extensions}^{j}$.
		Note, that for the reduction there is no computational difference between common release date and different deadlines and different release dates but common deadline.
		
		\item \label{alg:case:c-c} $j-1$ on cloud, $j$ on cloud: \\
		We $2$-approximates the resulting extensions, by precisely handling the communication to the server, but upperbounding the communication from the server.
		Repeat case \ref{alg:case:s-c} with the two following changes:\\
		For the checks before the problem conversion use $\comF{j-1,j'} + \procSF{j'} + \comF{j',j}$ and $\procCF{j'}$ instead of $\procSF{j'} + \comF{j',j}$ and $ \comF{j-1,j'} + \procCF{j'}$, respectively.
		Let $\jobs^{s'} \subseteq \jobs_{j}$ be the set of jobs actually put on the server in this step.
		Add $[\Delta^i + \max_{j' \in \jobs^{s'}} \comF{j-1,j'}, ~c\rightarrow c] = \sum_{j' \in \jobs^c} \procCF{j'} + V$ instead of $[\Delta^i,~c\rightarrow c] = \sum_{j' \in \jobs^c} \procCF{j'} + V$ to $\textsc{Extensions}^{j}$.
		We wait for the biggest communication delay to pass until we schedule the first job on the server.
		Note, that $\Delta^i + \max_{j' \in \jobs^{s'}} \comF{j',j} \leq 2 \Delta^i$ by construction.
	\end{enumerate}
\end{enumerate}

For every pair of a state $([t,j-1,\dynLoc]=\cost) \in \textsc{StateList}^{j-1}$ and $([\Delta t,\dynLocE{j-1}\rightarrow \dynLocE{j}] = \cost') \in \textsc{Extensions}^{j}$ with $\dynLoc = \dynLocE{j-1}$ add $[t+\Delta t,j,\dynLocE{j}]=\cost + \cost'$ to $\textsc{StateList}^{j}$.
After that process, for every triple $t,j,\dynLoc$ that has multiple states in  $\textsc{StateList}^{j}$ keep only the state with the lowest cost.
We can also discard states with $\cost > \budget$ and timestamp $t > 2\makeEstimate$.
Repeat this process with $j \rightarrow j+1$ until we computed $\textsc{StateList}^{n+1}$, simply move through that list and select the state with lowest timestamp $t$.
If there is no such state, there exist no schedule with makespan smaller or equal to \makeEstimate. 

\begin{lemma}
	Given a feasible \makeEstimate, the described procedure calculates a 2-approximation on the optimal makespan in time $poly(N,\makeEstimate)$
\end{lemma}
\begin{proof}
	We start by showing the approximation factor.
	Assume that we added $[\Delta^i,~c\rightarrow c] = \sum_{j' \in \jobs^c} \procCF{j'} + V$ instead of $[\Delta^i + \max_{j' \in \jobs^{s'}} \comF{j',j}, ~c\rightarrow c] = \sum_{j' \in \jobs^c} \procCF{j'} + V$ in step \ref{alg:case:c-c} above.
	That hypothetical algorithm would calculate a (possibly infeasible) solution with makespan $\makeAlg^{hypo} \leq \makeOpt$, since step \ref{alg:case:c-c} underestimates the needed time, and everything else is calculated precisely.
	The actual algorithm has makespan $\makeAlg \leq 2\makeAlg^{hypo}$ and therefore also $\makeAlg \leq 2\makeOpt$.
	
	We show the runtime of the algorithm by bounding the time needed for each iteration of: \ref{pro:extensions}. constructing state extensions $\textsc{Extensions}^{j}$, \ref{pro:combine}. combining the extensions with the previous $\textsc{StateList}^{j-1}$ and \ref{pro:cull}. culling duplicates from the resulting $\textsc{StateList}^{j}$.
	
	\begin{enumerate}
		\item \label{pro:extensions}
		For directly connected jobs $j-1$ and $j$ we can trivially calculate the 4 options in constant time.
		Therefore, we are interested in the runtime of steps \ref{alg:case:s-s}, \ref{alg:case:s-c}, \ref{alg:case:c-s} and \ref{alg:case:c-c} for some parallel subgraph with jobs $\jobs_j$.
		The steps get repeated for $\Delta^i$ in $\{0,\dots, \Delta^{max} \}$, where $\Delta^{max} < \makeEstimate$.
		The preprocessing in each iteration of all four steps, needs time linear in the size of $\jobs_j$.
		Using \algtardy in the steps needs time in $\mathcal{O}(\lvert \jobs_j \rvert \min\{\sum_{j' \in \jobs_{j} \setminus \jobs^c }\procSF{j'},\max_{j' \in \jobs_{j} \setminus \jobs^c } d_{j''}\}) \leq \mathcal{O}(\makeEstimate\cdot N^2) $.
		Overall we need time in $poly(\makeEstimate, N)$ to calculate $\textsc{Extensions}^{j}$, with $\lvert \textsc{Extensions}^{j}\rvert  \leq \mathcal{O}(\makeEstimate)$
		\item \label{pro:combine}
		$\textsc{StateList}^{j-1}$ contains at most $2\makeEstimate\cdot(n+2)\cdot 2$ (timestamp, job, location) different states (after the previous culling).
		We may simply bruteforce all possible combinations from $\textsc{StateList}^{j-1} \times \textsc{Extensions}^{j}$.
		Since both of these sets have at most $poly(\makeEstimate, N)$ elements, the resulting set $\textsc{StateList}^{j}$ also has polynomial size.
		\item \label{pro:cull}
		By culling states from $\textsc{StateList}^{j}$ we reduce it back to size at most $2\makeEstimate\cdot(n+2)\cdot 2$. It should be obvious, that we can identify duplicate states in polynomial time.
	\end{enumerate}
	
	Note that we iterate the above steps for each job $j \in \mset{1,\dots,n+1}$.
	Therefore we have a polynomial repetition of steps needing polynomial time.
	Note that we prevent exponential build-up in the state lists, by culling duplicates after each iteration.
\end{proof}

Now we have to scale our instance, such that our pseudo polynomial algorithm runs in proper polynomial time.
For that, we scale \makeEstimate and all $\procC$, \procS and \com by $\frac{N\varepsilon'}{ \makeEstimate}$ and round down to the next integer.
Then, we run our algorithm with the scaled values, but still use the unscaled \procC to calculate the \emph{value (cost)} of states, as those calculations only factor logarithmically in the runtime, a \procC exponential in the input size is fine.
The algorithm now needs time in $poly(N, \lfloor \frac{\makeEstimate \cdot N\varepsilon'}{ \makeEstimate} \rfloor ) \leq poly(N,\varepsilon')$ and finds a 2 approximation for the scaled instance (given a feasible \makeEstimate).
After scaling back up each job and communication delay might need up to $\frac{ \makeEstimate}{N\varepsilon'}$ additional time, delaying our whole schedule by at most $3N \cdot \frac{ \makeEstimate}{N\varepsilon'} \leq 3 \varepsilon' \makeEstimate$.
For $\varepsilon = 3 \varepsilon'$ and $\makeEstimate = \makeOpt$ our resulting schedule has a makespan of $\makeAlg \leq 2\makeOpt + \varepsilon \makeEstimate = (2+\varepsilon) \makeOpt$.
Via a binary search we can find such a $\makeEstimate$ by repeating our procedure at most $\log \sum_{j\in\jobs'}\procSF{j} $ times.
This concludes the proof of \Cref{the:exChainApprox}.

\begin{corollary}
	There is a polynomial algorithm for the deadline restrained cost minimization problem on extended chains, that finds a schedule with at most optimal cost, but a makespan of $(2+\varepsilon)\deadline$.
\end{corollary}

\subsection{Cases with FPTAS}
We reconsider the approximation result for three assumptions on the model which allow us to improve the result. 
Looking back at \Cref{the:exChainApprox}, we build an algorithm that would be an FPTAS if it were not for case \ref{alg:case:c-c} where we needed to double our time frame $\Delta^i$ to fit the unaccounted communication delay.
In the following part we will only describe how to approach that case, since everything else can stay as it was.

First we assume locally small delays in the parallel subgraphs, meaning that the smallest processing time in the subgraph is at least as big as the largest communication delay.
More precisely, for every $\jobs_{e}$ with $e=(j-1,j)$ it holds that 
\[ \min_{j'\in \jobs_{e}}\min \{ \procSF{j'}, \procCF{j'}  \} \geq \max_{j'\in \jobs_{e}} \max \{ \comF{(j-1,j')}, \comF{(j',j)}  \}. \] 
In this case only the first $j^\alpha$, and the last job $j^\omega$ to be processed on the server are actually affected by their communication delay, since all other delays fit in the time frame, where $j^\alpha$ and $j^\omega$ are processed.
After the preprocessing of a given $\Delta^i$, for each pair of jobs $j^\alpha, j^\omega \in \jobs_{j} \setminus \jobs^c$ with $j^\alpha \neq j^\omega$ fo the following:
Assume $j^\alpha, j^\omega$ are the first and last job to be processed on the server, respectively.
Add $j^\alpha$ and $j^\omega$ to $\jobs^s$.
Now create the \wntj instance as follows:
For every job $j' \in \jobs_{j} \setminus (\jobs^s \cup \jobs^c)$ create a job $j''$ with processing time $p_{j'} = \procSF{j'}$, deadline $d_{j''} = \Delta^i  - (\comF{j-1,j^\alpha} + \comF{j^\omega,j}) - \sum_{j' \in \jobs^s} \procSF{j'}$ and weight $w_{j''} = \procCF{j'}$.
Solve this problem with \algtardy, let $V$ be the cost of the solution and note $[\Delta^i,~c\rightarrow c]^{j^\alpha}_{j^\omega} = \sum_{j' \in \jobs^c} \procCF{j'} + V$.
After all ($\mathcal{O}(N^2)$) combinations have been tested, add the smallest $[\Delta^i,~c\rightarrow c]^{j^\alpha}_{j^\omega}$ to $\textsc{Extensions}^{j}$.

Secondly, we assume a constant upper bound $c_{max}$ on the communication delays inside parallel subgraphs.
More precisely, for every $\jobs_{e}$ with $e=(j-1,j)$ it holds that 
\[c_{max} \geq c(j-1,j') \text{ and } c_{max} \geq c(j',j)  . \]
Instead of brute forcing only a first and last job, we brute force the first and last $c_{max}$ time steps.
Trivially, jobs with $\procS = 0$ can be put on the server, and therefore there are at most $\mathcal{O}(N^c_{max}\cdot N^c_{max})$ combinations we have to work through.
The remaining part works analogously to the first case.

Lastly, we assume that each job produces some output, that has to be send to all of its direct successors in full, meaning that all outgoing communication delays of a job are equivalent.
More precisely, for every $\jobs_{e}$ with $e=(j-1,j)$ it holds that 
\[ \forall j',j'' \in \jobs_{e}:  c(j-1,j') = c(j-1,j'')  . \]
Here we can simply reuse the result from step \ref{alg:case:s-c}, but subtract $c(j-1,j')$ from the $\Delta^i$ used in the \wntj problem.
Since all $c(j-1,j')$ are equal, no job could be processed on the server in the first $c(j-1,j')$ time steps, and all jobs are available after those $c(j-1,j')$ time steps.

All these, in combination with the previously described scaling approach, lead to FPTAS results:
\begin{theorem}
	There is an FPTAS for the budget restrained makespan minimization problem on extended chains, if at least one of the following holds for every parallel subgraph $\jobs_{e}$ with $e=(j-1,j)$:
	\begin{enumerate}
		\item $\min_{j'\in \jobs_{e}}\min \{ \procSF{j'}, \procCF{j'}  \} \geq \max_{j'\in \jobs_{e}} \max \{ \comF{(j-1,j')}, \comF{(j',j)}  \}$
		\item $c_{max} \geq c(j-1,j') \text{ and } c_{max} \geq c(j',j)$
		\item $\forall j',j'' \in \jobs_{e}:  c(j-1,j') = c(j-1,j'')$
	\end{enumerate}
\end{theorem}

\subsection{Strong NP-Hardness of Scheduling Extended Chains}
As already noted, this problem is at least weakly NP-hard, following from \Cref{the:chainNPhard} as well as \Cref{the:fullyparallelNPhard}.
We show that this problem is actually strongly NP-hard, by giving a reduction from the strongly NP-hard $1\mid r_j\mid \sum w_j U_j$ problem \cite{lenstra1977complexity}.
As in \Cref{sec:prelimHardness} we use decision variants of the considered problems, resulting in results for both deadline restrained cost reduction and budget restrained makespan minimization.

\begin{theorem}
	The \modelnameextended problem is strongly NP-hard.
\end{theorem}
\begin{proof}
	$1\mid r_j\mid \sum w_j U_j$ is defined as follows:
	Given a set of jobs $\jobs = \{1,\dots,n\}$, each with processing time $p_j$, release date $r_j$, deadline $d_j$ and weight $w_j$, schedule the jobs (without preemption) on a single machine, such that the sum of weights of late jobs is smaller or equal to a given $b$ ($\sum w_j U_j \leq b$).
	A job $j$ is late ($U_j = 1$) if it finishes processing after $d_j$, $U_j = 0$ otherwise.
	
	Given an instance of $1\mid r_j\mid \sum w_j U_j$, create the following decision version of \modelnameextended.
	Note that we will substitute \enquote{an edge $(j,j')$ with communication delay $c(j,j')=k$} simply by \enquote{an edge $c(j,j')=k$} to keep this readable.
	As per definition create \source and \sink with $\procSF{\source} = \procSF{\sink} = 0$ and $\procCF{\source} = \procCF{\sink} = \infty$.
	Create jobs $j^{pre}$ and $j^{post}$ with $\procSF{j^{pre}} = \procSF{j^{post}} = \infty$ and $\procCF{j^{pre}} = \procCF{j^{post}} = 0$ and edges $c(\source,j^{pre})=0$ and $c(j^{post},\sink)=0$.
	Set $w^{max} = \max_{j \in \jobs} w_j$ and $d^{max} = \max_{j \in \jobs} d_j$. 
	For every $j\in \jobs$ create a job $j'$ with $\procSF{j'} = p_j$, $\procCF{j'} = w_j$ and edges $c(j^{pre},j')=r_j$, $c(j',j^{post})=w^{max} + d^{max} - d_j$.
	Set the deadline to $\deadline' = w^{max} + d^{max}$ and the budget $\budget' = \budget$.
	Trivially, in all schedules \source and \sink are scheduled on the server, $j^{pre}$ and $j^{post}$ on the cloud.
	Note that neither of these jobs contributes processing time to the resulting schedule.
	For better comprehension we give an example of the structure in \Cref{fig:reduExtendedchain}.
	
	\begin{figure}
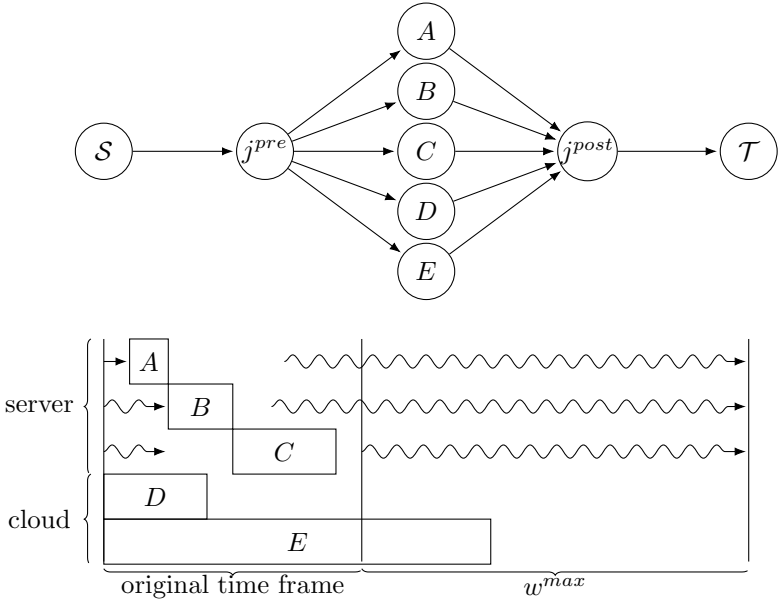

		\centering
		\tikz[>={Latex[length=1.5mm]}, shorten <= 0pt, shorten >= 1pt]{
			\pgfmathsetmacro{\w}{2.125}
			\pgfmathsetmacro{\h}{0.8}
			\pgfmathsetmacro{\nw}{0.75}
			

			\node[draw, circle, minimum size = \nw cm] (chain0) at (0,0) {$\source$};
			
			\node[draw, circle, inner sep=0, minimum size = \nw cm] (chain1) at (1*\w,0) {$j^{pre}$};
			\node[draw, circle, inner sep=0, minimum size = \nw cm] (chain2) at (2*\w,0) {$C$};
			\node[draw, circle, inner sep=0, minimum size = \nw cm] (chain3) at (3*\w,0) {$j^{post}$};
			\node[draw, circle, minimum size = \nw cm] (chain4) at (4*\w,0) {$\sink$};
			
			\foreach \x/\y in {0/1,1/2,2/3,3/4}
			{
				\draw[->] (chain\x) -- (chain\y);
			}
			
			\node[draw, circle, inner sep=0, minimum size = \nw cm] (parallel11) at (2*\w,2*\h) {$A$};
			\node[draw, circle, inner sep=0, minimum size = \nw cm] (parallel12) at (2*\w,1*\h) {$B$};
			\node[draw, circle, inner sep=0, minimum size = \nw cm] (parallel13) at (2*\w,-1*\h) {$D$};
			\node[draw, circle, inner sep=0, minimum size = \nw cm] (parallel14) at (2*\w,-2*\h) {$E$};

			\foreach \x in {1,2,3,4}
			{
				\draw[->] (chain1) -- (parallel1\x);
				\draw[->] (parallel1\x) -- (chain3);
			}

			\pgfmathsetmacro{\yoffset}{-2.5}
			\pgfmathsetmacro{\ww}{1.7}
			\pgfmathsetmacro{\hh}{0.6}
			
			\draw[-] (0,0+\yoffset) -- (0,-5*\hh+\yoffset);
			\draw[-] (2*\ww,0+\yoffset) -- (2*\ww,-5*\hh+\yoffset);
			\draw[-] (5*\ww,0+\yoffset) -- (5*\ww,-5*\hh+\yoffset);
			
			\draw [->, decorate, decoration = {snake, post length = 4}] (0*\ww ,-0.5*\hh+\yoffset) -- (0.2*\ww ,-0.5*\hh+\yoffset);
			\draw [draw=black] (0.2*\ww ,0+\yoffset) rectangle ++(0.3*\ww,-1*\hh);
			\draw [->, decorate, decoration = {snake, post length = 4}] (1.4*\ww ,-0.5*\hh+\yoffset) -- (5*\ww ,-0.5*\hh+\yoffset);
			\node[] at (0.2*\ww + 0.15*\ww ,0+\yoffset-0.5*\hh) {$A$};
			
			\draw [->, decorate, decoration = {snake, post length = 4}] (0*\ww ,-1.5*\hh+\yoffset) -- (0.5*\ww ,-1.5*\hh+\yoffset);
			\draw [draw=black] (0.5*\ww ,-1*\hh+\yoffset) rectangle ++(0.5*\ww,-1*\hh);
			\draw [->, decorate, decoration = {snake, post length = 4}] (1.3*\ww ,-1.5*\hh+\yoffset) -- (5*\ww ,-1.5*\hh+\yoffset);
			\node[] at (0.5*\ww + 0.25*\ww ,-1*\hh+\yoffset-0.5*\hh) {$B$};
			
			\draw [->, decorate, decoration = {snake, post length = 4}] (0*\ww ,-2.5*\hh+\yoffset) -- (0.5*\ww ,-2.5*\hh+\yoffset);
			\draw [draw=black] (1*\ww ,-2*\hh+\yoffset) rectangle ++(0.8*\ww,-1*\hh);
			\draw [->, decorate, decoration = {snake, post length = 4}] (2*\ww ,-2.5*\hh+\yoffset) -- (5*\ww ,-2.5*\hh+\yoffset);
			\node[] at (1*\ww + 0.4*\ww ,-2*\hh+\yoffset-0.5*\hh) {$C$};
			
			\draw [decorate, decoration = {brace, mirror}] (-0.1*\ww,0+\yoffset) --  (-0.1*\ww,-3*\hh+\yoffset);
			\node[] at (-0.9*\ww + 0.4*\ww ,-1*\hh+\yoffset-0.5*\hh) {server};
			
			\draw [draw=black] (0*\ww ,-3*\hh+\yoffset) rectangle ++(0.8*\ww,-1*\hh);
			\node[] at (0*\ww + 0.4*\ww ,-3*\hh+\yoffset-0.5*\hh) {$D$};
			\draw [draw=black] (0*\ww ,-4*\hh+\yoffset) rectangle ++(3*\ww,-1*\hh);
			\node[] at (0*\ww + 1.5*\ww ,-4*\hh+\yoffset-0.5*\hh) {$E$};
			
			\draw [decorate, decoration = {brace, mirror}] (-0.1*\ww,-3*\hh+\yoffset) --  (-0.1*\ww,-5*\hh+\yoffset);
			\node[] at (-0.9*\ww + 0.4*\ww ,-3.5*\hh+\yoffset-0.5*\hh) {cloud};

			\draw [decorate, decoration = {brace, mirror}] (0*\ww,-5.1*\hh+\yoffset) --  (2*\ww,-5.1*\hh+\yoffset);
			\node[] at (1*\ww ,-5.5*\hh+\yoffset) {original time frame};
			\draw [decorate, decoration = {brace, mirror}] (2*\ww,-5.1*\hh+\yoffset) --  (5*\ww,-5.1*\hh+\yoffset);
			\node[] at (3.5*\ww ,-5.5*\hh+\yoffset) {$w^{max}$ };
		}
		\caption{Schematic example of resulting \modelnameextended problem for 5 jobs, squiggly arrows represent communication delays and model release dates and deadlines.}
		\label{fig:reduExtendedchain}
	\end{figure}
	
	It remains to show, that there is a schedule with $\sum w_j U_j \leq b$ for the original $1\mid r_j\mid \sum w_j U_j$ problem, iff
	there is a schedule with cost $\leq b'$ and makespan $\leq \deadline'$ for the \modelnameextended problem.
	
	Assume that there is a schedule with $\sum w_j U_j \leq \budget$.
	We can partition the jobs into two sets $\jobs^{early}$ and $\jobs^{late}$, which contain all jobs that are on time or late, respectively.
	Place all jobs that correspond to a job from $\jobs^{late}$ on the cloud and start them immediately.
	All of them finish before $\deadline' = w^{max} + d^{max}$, since $w^{max} \geq \procCF{j'}$.
	Place all remaining jobs ($\jobs^{early}$) on the server and let them start at the same time as in the original schedule.
	Since no job starts before its release date no communication delay is violated in the new schedule.
	Since all jobs from $\jobs^{early}$ end before their deadline, no communication delay hinders us from scheduling $j^{post}$ and \sink at $\deadline' = \Delta^{max} + d^{max}$.
	The cost of that schedule is equal to the value of $\sum w_j U_j$ in the original schedule and therefore $\leq \budget$.
	One can confirm that the other direction works analogously by keeping the schedule of jobs on the cloud intact, and simply processing all jobs from the cloud after that schedule in any order.
	
\end{proof}

With argumentation similar to the reduction above, one can show that the $1\mid r_j\mid \sum w_j U_j$ problem is embedded in step \ref{alg:case:c-c} of this chapter's algorithm.
This leads to the observation, that we might be able to use approximation results for $1\mid r_j\mid \sum w_j U_j$ to improve our handling of that case.
Sadly, to the best of our knowledge, no approximation algorithms with a provable approximation factor are known for this problem.
There are however practical algorithms, which have been tested empirically.
Used approaches contain mixed integer programming \cite{DBLP:journals/eor/Detienne14}, genetic algorithms \cite{DBLP:journals/eor/SevauxD03} and branch-and-bound algorithms \cite{DBLP:journals/eor/MHallahB07}.
For more information we again refer to \cite{adamu2014survey}.

\section{Constant Cardinality Source and Sink Dividing Cut}\label{sec:constant_width_FPTAS}

We introduce the concept of a \emph{maximum cardinality source and sink dividing cut}.
For $G=(\jobs,E)$, let $\jobs_\source$ be a subset of jobs, such that $\jobs_\source$ includes \source and there are no edges $(j,k)$ with $j\in \jobs \setminus \jobs_\source$ and $k \in \jobs_\source$.
In other words, in a running schedule $\jobs_\source$ and $\jobs \setminus \jobs_\source$, could represent already processed jobs and still to be processed jobs respectively.
Denote by $\jobs_\source^G$ the set of all such sets $\jobs_\source$.
We define
\[ \dynConc := \max_{\jobs_\source \in \jobs_\source^G} \mid  \{ (j,k)\in E ~\mid ~ j\in \jobs_\source \land k\in \jobs \setminus \jobs_\source  \} \mid  , \]
the maximum number of edges between any set $\jobs_\source$ and $\jobs \setminus \jobs_\source$ in $G$.
In a series-parallel task graph \dynConc is equal to the maximum anti-chain size of the graph.

In this chapter we discuss how to solve or approximate \modelname problems with a constant size \dynConc, but otherwise arbitrary task graphs.
We first consider the deadline confined cost minimization, in \Cref{the:dynResultBudget} we show how to adapt this to the budget confined makespan minimization. 
We give a dynamic program to optimally solve instances of \modelname with arbitrary task graphs.
At first we will not confine the algorithm to polynomial time.
Consider a given problem instance with $G=(\jobs, E)$, its source \source and sink \sink, processing times $\procSF{j}$ and $\procCF{j}$ for each $j \in \jobs$, communication delays $\comF{i,j}$ for each $(i,j)\in E$ and a deadline \deadline.

We define intermediate states of a (running) schedule, as the states of our dynamic program (see \cref{fig:dynState}).
Such a state contains two types of variables.
First we have two global variables, the timestamp \dynTimestamp and the number of time steps the server has been unused \dynFreeServer.
In other words, the server has not finished processing a job since $\dynTimestamp - \dynFreeServer$.
The second type is defined per \emph{open edge}.
\begin{figure}
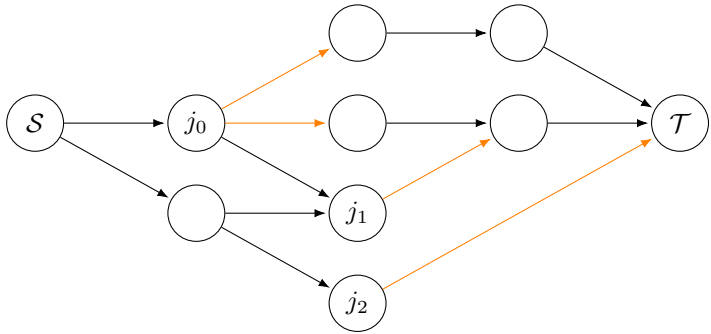

	\centering
	\tikz[>={Latex[length=1.5mm]}, shorten <= 0pt, shorten >= 1pt]{
		\pgfmathsetmacro{\w}{2.125}
		\pgfmathsetmacro{\h}{0.6}
		\pgfmathsetmacro{\nw}{0.75}
		

		\node[draw, circle, minimum size = \nw cm] (chain0) at (0,0) {\source};
		
		\node[draw, circle, inner sep=0, minimum size = \nw cm] (chain1) at (1*\w,0) {$j_0$};
		\node[draw, circle, inner sep=0, minimum size = \nw cm] (chain2) at (2*\w,0) {};
		\node[draw, circle, inner sep=0, minimum size = \nw cm] (chain3) at (3*\w,0) {};
		\node[draw, circle, minimum size = \nw cm] (chain4) at (4*\w,0) {\sink};
		
		\foreach \x/\y in {0/1,2/3,3/4}
		{
			\draw[->] (chain\x) -- (chain\y);
		}
		\draw[->, color=orange] (chain1) -- (chain2);
		
		\node[draw, circle, inner sep=0, minimum size = \nw cm] (A) at (2*\w,2*\h) {};
		\node[draw, circle, inner sep=0, minimum size = \nw cm] (B) at (3*\w,2*\h) {};
		\node[draw, circle, inner sep=0, minimum size = \nw cm] (C) at (2*\w,-2*\h) {$j_1$};
		\node[draw, circle, inner sep=0, minimum size = \nw cm] (D) at (1*\w,-2*\h) {};
		\node[draw, circle, inner sep=0, minimum size = \nw cm] (E) at (2*\w,-4*\h) {$j_2$};
		
		\draw[->] (chain0) -- (D);
		\draw[->] (D) -- (C);
		\draw[->] (D) -- (E);
		\draw[->, color=orange] (E) -- (chain4);
		\draw[->, color=orange] (C) -- (chain3);
		\draw[->, color=orange] (chain1) -- (A);
		\draw[->] (A) -- (B);
		\draw[->] (B) -- (chain4);
		\draw[->] (chain1) -- (C);
	}
	\label{fig:dynState}
	\caption{Example state of a running schedule, \emph{open edges} are orange, $\dynLocE{j_i}$ and $\dynFreeE{j_i}$ kept for $j_0$, $j_1$ and $j_2$.}
\end{figure}
An open edge is a $e=(j,k)$ where $j$ has already been processed, but $k$ has not.
For each such edge add the variables $e=(j,k)$ (the edge itself), $\dynLocE{j}\in \{s,c\}$ denoting if $j$ was processed on the server ($s$) or the cloud ($c$) and \dynFreeE{j} denoting the number of time steps that have passed since $j$ finished processing.
If a job $j$ is contained in multiple open edges, $\dynLocE{j}$ and \dynFreeE{j} are still only included once.
Write the state as $[ \dynTimestamp, \dynFreeServer, e^1=(j^1,k^1), \dynLocE{j^1}, \dynFreeE{j^1}, \dots,  e^m=(j^m,k^m), \dynLocE{j^m}, \dynFreeE{j^m}]$, where $e^1, \dots, e^m$ denote all open edges.
Note here, that there is information that we purposefully drop from a state: the completion time and location of every processed job without open edges, as those are not important for future decisions anymore.
There might be multiple ways to reach a specific state, but we only care about the minimum possible cost to achieve that state, which is the \emph{value} of the state.

We iteratively calculate the value of every reachable state with $\dynTimestamp = 0, 1, 2,\dots$.
We start with the trivial state $[\dynTimestamp = 0, \dynFreeServer = 0,  e^1, \dots, e^m, \dynLocE{\source}=s, \dynFreeE{\source}=0] = 0$, where $e^1, \dots, e^m \in E$ with $e^i = (\source, j)$.
This state forms the beginning of our \emph{(sorted) state list}.
We keep this list sorted in an ascending order of state values (costs) at all times.
We exhaustively calculate every state that is reachable during a specific time step, given the set of states reachable during the previous time step.
Intuitively, we try every possible way to \enquote{fill up} the still undefined time windows \dynFreeServer and \dynFreeE{j}.

Finally, we give the actual dynamic program in \Cref{alg:dynProg}.
After the dynamic program finished, we iterate through the state list one last time and take the first state $[\dynTimestamp=\deadline, \dynFreeServer]$.
The value of that state is the minimum cost possible to schedule $G$ in time $\deadline$.
One can easily adapt this procedure to also yield such a schedule, by keeping a list of all processed jobs per state containing their location and completion time.

\begin{lemma}
	\label{lem:dynRuntime}
	\dynProgName's runtime is bounded in $\mathcal{O}(d^{2\dynConc+3} \cdot n^{2\dynConc+1} )$.
\end{lemma}
\begin{proof}
	At any point there are a maximum of $\mathcal{O}(\deadline \cdot (\deadline \cdot n)^\dynConc)$ states in the state list.
	For every \dynTimestamp we look at every state.
	Since we never insert a state in front of the state we are currently inspecting (costs can only increase), this traverses the list exactly once.
	For each of those states we calculate every possible successor, of which there are $\mathcal{O}(\dynConc)$ and traverse the state list an additional time to correctly insert or update the state.
	We iterate from $\dynTimestamp = 0$ to $\deadline$ and therefore get a runtime of: $\mathcal{O}( \deadline \cdot ( (\deadline \cdot (\deadline \cdot n)^\dynConc) \cdot \dynConc \cdot (\deadline \cdot (\deadline \cdot n)^\dynConc) )) = \mathcal{O}(d^3 \cdot n \cdot(\deadline \cdot n)^{2\dynConc}) \leq \mathcal{O}(d^{2\dynConc+3} \cdot n^{2\dynConc+1} )$. 
\end{proof}

\alglanguage{pseudocode}
\begin{algorithm}[H]
	\caption{\dynProgName: Dynamic Program for General Graphs}
	\label{alg:dynProg}
	\begin{algorithmic}[1]
		\State initialize state list $SL$ with start state (as defined above)
		\ForAll {$\dynState \in SL$}
		\State let $\jobs^\dynState$ be the set of all jobs that are endpoints in open edges from \dynState
		\ForAll {$j \in \jobs^\dynState$}
		\If {$\forall (k,j)\in E: (k,j)$ also open edge in $\dynState$}
		\State \emph{can $j$ be processed on the server?}
		\If {$\dynFreeServer \geq \procSF{j}$}
		\State $jFits \gets TRUE$
		\ForAll {$(k,j)\in E$}
		\If {$\dynLocE{k}=s \land \dynFreeE{k} < \procSF{j}$ or $\dynLocE{k}=c \land \dynFreeE{k} < \procSF{j} + \comF{k,j}$}
		\State $jFits \gets FALSE$
		\EndIf
		\EndFor
		\If {$jFits = TRUE$}
		\State \emph{calculate resulting state $\dynState'$, value equal to $\dynState$}
		\ForAll {$(k,j)\in E$}
		\State remove $(k,j)$ from $\dynState'$
		\If {$j$ is last open successor of $k$}
		\State remove $\dynFreeE{k}$ and $\dynLocE{k}$ from $\dynState'$
		\EndIf
		\State add $\dynFreeE{j}=0$, $\dynLocE{j}=s$ and all new open edges to $\dynState'$
		\EndFor
		\If {$\dynState' \in SL $} 
		\State update value of $\dynState'$ in $SL$ if new value lower
		\State then move $\dynState'$ to correct position in $SL$
		\Else
		\State add $\dynState'$ to correct position in $SL$ (always after \dynState)
		\EndIf
		\EndIf
		\EndIf
		\State \emph{can $j$ be processed on the cloud?}
		\State \emph{analogously to the previous case, cost value of $\dynState'$ increased by $\procCF{j}$}
		\EndIf
		\EndFor
		\EndFor
		\State \emph{check end condition}
		\If {a state $[\dynTimestamp=\deadline, \dynFreeServer] \in SL$}
		\State return lowest value of such states
		\EndIf
		\If {$\dynTimestamp < \deadline $}
		\State \emph{move from $\dynTimestamp$ to $\dynTimestamp+1$}
		\ForAll {each $\dynState \in SL$}
		\State increase $\dynTimestamp$, $\dynFreeServer$ and each $\dynFreeE{j}$ in \dynState by $1$
		\EndFor
		\State Back to step 2
		\EndIf
	\end{algorithmic}
\end{algorithm}
\newpage

\subsection{Rounding the Dynamic Program}
We use a rounding approach on \dynProgName to get a program that is polynomial in $n = \mid \jobs\mid $, given that $\dynConc$ is constant.
We scale $\deadline$, $\com$, $\procC$, and $\procS$ by a factor $\scale := \frac{\varepsilon \cdot \deadline}{2n}$.
Denote by $\scaled{\deadline} := \lceil \frac{ \deadline }{\scale} \rceil \leq \frac{2n}{\varepsilon} + 1 $, $\scaledProcS{j} := \lfloor\frac{ \procSF{j} }{\scale}\rfloor $, $\scaledProcC{j} := \lfloor\frac{ \procCF{j} }{\scale}\rfloor $ and $\scaledComF{x} := \lfloor\frac{\comF{x}}{\scale} \rfloor $.
Note here, that we round up \deadline but everything else down.
We run the dynamic program with the rounded values, but still calculate the cost of a state with the original unscaled values.

We transform the output $\pi'$ to the unscaled instance, by trying to start every job $j$ at the same (scaled back up) point in time as in the scaled schedule.
Since we rounded down, there might now be points in the schedule where a job $j$ can not start at the time it is supposed to.
This might be due to the server not being free, a parent node of $j$ that has not been fully processed or an unfinished communication delay.
We look at the first time this happens and call the mandatory delay on $j$ $\Delta$ and increase the start time of every remaining job by $\Delta$.
Repeat this process until all jobs are scheduled.
We introduce no new conflicts with this procedure, since we always move everything together as a block.
Call this new schedule $\pi$.

\begin{theorem}
	\label{the:dynResult}
	Assuming a constant number $\dynConc$ \dynProgName combined with the scaling technique finds a schedule $\pi$ with at most optimal cost and a makespan ~$\leq (1+\varepsilon)\cdot \deadline$ in time $poly(n,\frac{1}{\varepsilon})$, for any $\varepsilon>0$.
\end{theorem}
\begin{proof}
	We start by proving the runtime of our algorithm.
	We can scale the instance in polynomial time, this holds for both scaling down and scaling back up.
	The dynamic program now takes time in $\mathcal{O}(\scaled{d}^{2\dynConc+3} \cdot n^{2\dynConc+1} )$, where $\scaled{\deadline}\leq \frac{2n}{\varepsilon} + 1 $.
	Since $\dynConc$ is constant this results in an dynamic program runtime in $poly(n,\frac{1}{\varepsilon})$.
	In the end we transform the schedule as described above, for that we go trough the schedule once and delay every job no more than $n$ times.
	Trivially, this can be done in polynomial time as well.
	
	Secondly we show that the makespan of $\pi$ is at most $(1+\varepsilon)\cdot \deadline$.
	Every valid schedule for the unscaled problem is also valid in the scaled problem, meaning that there is no possible schedule we overlook due to the scaling.
	In the other direction this might not hold.
	First, while scaling everything down we rounded the deadline up.
	This means, that scaled back we might actually work with a deadline of up to $\deadline + \scale$.
	Secondly, we had to delay the start of jobs to make sure that we only start jobs when it is actually possible.
	In the worst case we delay the sink \sink a total of $n-2$ times, once for every job other than \source and \sink.
	Each time we delay all remaining jobs we can bound the respective $\Delta < 2 \cdot \scale$.
	This is due to the fact that each of the delaying options can not delay by more than \scale (as that is the maximum timespan not regarded in the scaled problem) and only a direct predecessor job and the communication from it needing longer can coincide to a non-parallel delay.
	Taking both of these into account, a valid schedule for the scaled problem might use time up to 
	\[\deadline + \scale + (n-2)\cdot (2\scale) \leq \deadline + 2n\scale = (1+\varepsilon)\cdot  \deadline \]
	in the unscaled instance.
	
	Lastly, we take a look at the cost of $\pi$.
	While rounding, we did not change the calculation of a states value, and with every valid schedule of the unscaled instance being still valid in the scaled instance we can conclude that the cost of $\pi$ is smaller or equal to an optimal solution of the original problem. 
\end{proof}
\begin{theorem}
	\label{the:dynResultBudget}
	\dynProgName combined with the scaling technique and a binary search over the deadline yields an FPTAS for the cost budget makespan problem, for graphs with a constant number \dynConc.
\end{theorem}
\begin{proof}
	\Cref{the:dynResult} can be adapted to solve this, assuming that we know a reasonable makespan estimate of an optimal solution to use in our scaling factor.
	During the algorithm discard any state with costs bigger than the budget and terminate when the first state $[\dynTimestamp, \dynFreeServer]$ is reached.
	The $\dynTimestamp$ gives us the makespan.
	
	Using a makespan estimate that is too big will lead to a rounding error that is not bounded by $\varepsilon \cdot \makeOpt$, a too small estimate might not find a solution.
	To solve this, we start with an estimate that is purposefully large.
	Let $\deadline^{max} = \sum_{j\in \jobs}\procSF{j}$ be the sum over all processing times on the server.
	There is always a schedule with 0 costs and makespan $\deadline^{max}$.
	We run our algorithm with the scaling factor $\scale^0 := \frac{\varepsilon \cdot \deadline^{max}}{4n}$.
	Iteratively repeat this process with scaling factor $\scale^i = \frac{1}{2^i}\scale^0$ for increasing $i$ starting with $1$.
	At the same time half the original deadline estimate in each step, which leads to $\scaled{\deadline}$, and therefore the runtime, to stay the same in each iteration.
	End the process when the algorithm does not find a solution for the current $i$ and deadline estimation.
	This infers that there is no schedule with the wanted cost budget and a makespan smaller or equal to $\frac{1}{2^i}\deadline^{max}$ (in the unscaled instance), therefore $\frac{1}{2^i}\deadline^{max} < \makeOpt$.
	We look at the result of the previous run $i-1$: The scaled result was optimal, therefore the unscaled version has a makespan of at most
	\begin{align}
	\makeAlg &\leq \makeOpt + 2n \cdot \scale^{i-1}\\
	&= \makeOpt + 2n \cdot \frac{1}{2^{i-1}} \cdot \frac{\varepsilon \cdot \deadline^{max}}{4n}\\
	&= \makeOpt +  \varepsilon \cdot  \frac{1}{2^{i}}\deadline^{max} \leq (1+\varepsilon) \makeOpt.
	\end{align}
	
	It should be easy to infer from \Cref{lem:dynRuntime} that each iteration of this process has polynomial runtime.
	Combined with the fact that we iterate at most $\log \deadline^{max}$ times we get a runtime that is in $poly(n,\frac{1}{\varepsilon})$. 
\end{proof}
\begin{remark}
	The results of this chapter work, as written, for a constant \dynConc.
	Note here, that for series parallel digraphs, this is equivalent to a constant anti-chain size.
	The algorithms can also be adapted to work on any graph with constant anti-chain size, if the communication delays are bounded by some constant or are \emph{locally small}.
	Delays are locally small, if for every $(j,k)\in E$, $c(j,k)$ is smaller or equal than every $\procCF{k'}$, $\procSF{k'}$, $\procCF{j'}$ and $\procSF{j'}$, where $k'$ is every direct successor of $j$ and $j'$ every direct predecessor of $k$ \cite{DBLP:conf/esa/MohringSS96}.
\end{remark}

\section{Strong NP-Hardness}\label{sec:strong_hardness}

In this section, we consider more involved reductions then in \cref{sec:prelim} in order to gain a better understanding for the complexity of the problem.
First, we show that a classical result due to Lenstra and Rinnooy Kan \cite{DBLP:journals/ior/LenstraK78} can be adapted to prove that already the variant of \modelname without communication delays and processing times equal to one or two is NP-hard.
This already implies strong NP-hardness.
Remember that we did show in \cref{sec:prelim} that \modelname without communication delays and with unit processing times can be solved in polynomial time.
Hence, it seems natural to consider the problem variant with unit processing times and communication delays.
We prove this problem to be NP-hard as well via an intricate reduction from $3SAT$ that can be considered the main result of this section.
Lastly, we show that the latter reduction can be easily modified to get a strong inapproximability result regarding the general variant of \modelname and the cost objective.

\subsection{No Delays and Two Sizes}

We show strong hardness for the case without communication delays and $\procC(j),\procS(j)\in\mset{1,2}$ for each job $j$.
The reduction is based on a classical result due to Lenstra and Rinnooy Kan \cite{DBLP:journals/ior/LenstraK78}.

Let $G=(V,E)$, $k$ be a clique instance with $\mid E\mid  > \binom{k}{2}$, and let $n = \mid V\mid $ and $m = \mid E\mid $.
We construct an instance of the cloud server problem in which the communication delays all equal zero and both the deadline and the cost bound is $2n + 3m$.
There is one vertex job $J(v)$ for each node $v\in V$ and one edge job $J(e)$ for each edge $e\in E$ and $J(\mset{u,v})$ is preceded by $J(u)$ and $J(v)$.
The vertex jobs have size $1$ and the edge jobs size $2$ both on the server and on the cloud.

Furthermore there is a dummy structure. 
First, there is a chain of $2n + 3m$ many jobs called the anchor chain.
The $i$-th job of the anchor chain is denoted $A(i)$ for each $i\in\mset{0,\dots2n+3m - 1}$ and has size $1$ on the cloud and size $2$ on the server.
Next, there are gap jobs each of which has size $1$ both on the server and the cloud. 
Let $k^* = \binom{k}{2}$ and $v\prec w$ indicate that an edge from $v$ to $w$ is included in the task graph.
There are four types of gap jobs, namely 
$G(1,i)$ for $i\in \mset{0,\dots k-1}$ with edges $ A(2i) \prec G(1,i) \prec A(2(i+1))$, 
$G(2,i)$ for $i\in \mset{0,\dots k^* - 1}$ with $A(2k + 3i + 1) \prec G(2,i) \prec A(2k + 3(i+1))$, 
$G(3,i)$ for $i\in \mset{0,\dots (n-k)-1}$ with $A(2k + 3k^* + 2i) \prec G(3,i) \prec A(2k + 3k^* +2(i+1))$, 
and $G(4,i)$ for $i\in \mset{0,\dots (m - k^*) - 1}$ with $A(2n + 3k^* + 3i + 1) \prec G(4,i) \prec A(2n + 3k^* +3(i+1))$ for $i< (m - k^*) -1$ and $A(2n + 3m - 2) \prec G(4,(m-k^*) - 1) $.
Lastly, there are the source and the sink which precedes or succeeds all of the above jobs, respectively.
\begin{lemma}
	There is a $k$-clique, if and only if there is a schedule with length and cost at most $2n + 3m$.
\end{lemma}
\begin{proof}
	First note that in a schedule with deadline $2n + 3m + 1$ the anchor chain has to be scheduled completely on the cloud.
	If the schedule additionally satisfies the cost bound, all the other jobs have to be scheduled on the server.
	Furthermore, for the gap and anchor chain jobs there is only one possible time slot due to the deadline.
	In particular, $A(i)$ starts at time $i$, $G(1,i)$ at time $2i+1$, $G(2,i)$ at time $2k + 3i + 2$, $G(3,i)$ at time $2k + 3k^* + 2i + 1$, and $G(4,i)$ at time $2n + 3k^* + 3i + 2$.
	Hence, there are $k$ length $1$ slots positioned directly before the $G(1,i)$ jobs left on the server, as well as, $k^*$ length $2$ slots directly before the $G(2,i)$ jobs, $n-k$ length $1$ slots directly before the $G(3,i)$ jobs, and $m-k^*$ length $2$ slots directly before the $G(2,i)$ jobs (see also \cref{fig:12ReductionDummy}).
	The $m$ edge jobs have to be scheduled in the length $2$ slots, and hence the vertex jobs have to be scheduled in the length $1$ slots.
	
	$\implies$:
	Given a $k$-clique, we can position the $k$ clique vertices in the first $k$ length $1$ slots, the corresponding $k^*$ edges in the first length $2$ slots, the remaining vertex jobs in the remaining length $1$ slots, and the remaining edge jobs in the remaining length $2$ slots.
	
	$\impliedby$:
	Given a feasible schedule, the vertices corresponding to the first length $1$ slots have to form a clique.
	This is the case, because there have to be $k^*$ edge jobs in the first length $2$ slots and all of their predecessors are positioned in the first length $1$ slots.
	This is only possible if these edges are the edges of a $k$-clique.
\end{proof}
\begin{figure}
	\begin{tikzpicture}[>={Latex[length=1.1mm]}, shorten <= 0pt, shorten >= 0.5pt, scale = 1.0]
	\pgfmathsetmacro{\c}{0}
	\pgfmathsetmacro{\s}{-0.6}
	\pgfmathsetmacro{\w}{0.32}
	\pgfmathsetmacro{\ws}{1.35}
	
	\node at (0,\c)[above right, yshift = -5.0*\w] {Cloud:};
	\node at (0,\s)[above right, yshift = -5.0*\w] {Server:};
	
	\foreach \x in {0,1,2,5,6,7,8,9,10,13,14,15,16,17,18,21,22,23,24,25,26,29,30,31}{
		\draw (\ws + \x*\w, \c) rectangle ++(\w, \w);
	}
	\foreach \x in {4,12,20,28}{
		\node at (\ws + \x*\w, \c + 0.5*\w) {$\dots$};
		\node at (\ws + \x*\w, \s + 0.5*\w) {$\dots$};
	}
	\foreach \x in {1,6,9,15,17,22,25}{
		\draw (\ws + \x*\w, \s) rectangle ++(\w, \w);
		\draw[->] (\ws + \x*\w - 0.5*\w, \c) to  (\ws + \x*\w + 0.2*\w, \s + \w);
		\draw[->] (\ws + \x*\w + 0.8*\w, \s + \w) to (\ws + \x*\w + 1.5*\w, \c);
	}
	\draw (\ws + 31*\w, \s) rectangle ++(\w, \w);
	\draw[->] (\ws + 31*\w - 0.5*\w, \c) to  (\ws + 31*\w + 0.2*\w, \s + \w);
	
	\foreach \x/\y in {0/0,7/2k,16/2k+3k^*,23/2n+3k^*,32/2n+3m}{
		\draw[thick] (\ws + \x*\w, \c + 1.5*\w) -- (\ws + \x*\w, \s - 0.5*\w);
		\node at (\ws + \x*\w, \c + 2.1*\w) {\scriptsize $\y$};
	}
	
	\draw[decorate,decoration={brace,amplitude=4pt,raise=1pt,mirror},yshift=-0pt] (\ws + 0*\w, \s - 0.5*\w) -- (\ws + 6.9*\w, \s - 0.5*\w) node [midway,yshift=-9pt]{\scriptsize $k$ size $1$ slots};
	\draw[decorate,decoration={brace,amplitude=4pt,raise=1pt,mirror},yshift=-0pt] (\ws + 7.1*\w , \s - 0.5*\w) -- (\ws + 15.9*\w, \s - 0.5*\w) node [midway,yshift=-9pt]{\scriptsize $k^*$ size $2$ slots};
	\draw[decorate,decoration={brace,amplitude=4pt,raise=1pt,mirror},yshift=-0pt] (\ws + 16.1*\w , \s - 0.5*\w) -- (\ws + 22.9*\w, \s - 0.5*\w) node [midway,yshift=-9pt]{\scriptsize $n-k$ size $1$ slots};
	\draw[decorate,decoration={brace,amplitude=4pt,raise=1pt,mirror},yshift=-0pt] (\ws + 23.1*\w , \s - 0.5*\w) -- (\ws + 31.9*\w, \s - 0.5*\w) node [midway,yshift=-9pt]{\scriptsize $m-k^*$ size $2$ slots};
	\end{tikzpicture}
	\caption{The dummy structure for the reduction from the clique problem to a special case of \modelname. Time flows from left to right, the anchor chain jobs are positioned on the cloud, and the gap jobs on the server.}
	\label{fig:12ReductionDummy}
\end{figure}
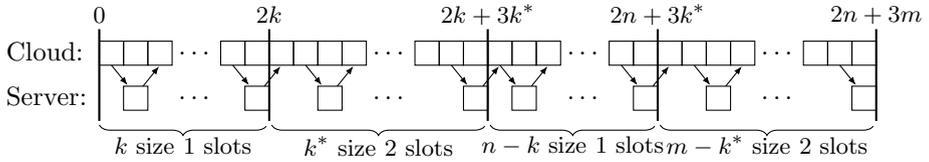

\noindent Hence, we have:
\begin{theorem}
	\label{the:12StrongNP}
	The \modelname problem with job sizes 1 and 2 and without communication delays is strongly NP-hard.
\end{theorem}
In the above reduction the server and the cloud machines are unrelated relative to each other due to different sizes of the anchor chain jobs.
However, it is easy to see that the reduction can be modified to a uniform setting where the cloud machines have speed $2$ and the server speed $1$.
If we allow communication delays, even identical machines can be achieved.
%

\subsection{Unit Size and Unit Delay}

We consider a unit time variant of our model in which all $\procC = \procS = 1$ and all $\com = 1$.
Note here, that this also implies that the server and the cloud are identical machines (the cloud still produces costs, while the server does not).
As usual for reductions we look at the decision variant of the problem: Is there a schedule with cost smaller or equal to \budget while adhering to the deadline \deadline.

\begin{theorem}
	\label{the:unitStrongNP}
	The \modelnameunit problem is strongly NP-hard.
\end{theorem}

We give a reduction $3SAT \leq_p \modelnameunit$.
Let $\phi$ be any boolean formula in 3-CNF, denote the variables in $\phi$ by $\mathcal{X}=\{x_1,x_2,\dots,x_m\}$ an the clauses by $\mathcal{C}=\{\clause{1}, \clause{2}, \dots, \clause{n}\}$.
Before we define the reduction formula we want to give an intuition and a few core ideas used in the reduction.

The main idea is that we ensure that nearly everything has to be processed on the cloud, there are only a few select jobs that can be handled by the server.
For each variable there will be two jobs, of which one can be processed on the server, the selection will represent an assignment.
For each clause there will be a job per literal in that clause, only one of which can be processed on the server, and only if the respective variable job is 'true'.
Only if for each variable and for each clause one job is handled by the server the schedule will adhere to both the cost and the time limits.

A core technique of the reduction is the usage of an anchor chain.
An anchor chain of length $l$ consists of two chains of the same length $l := \deadline - 2$, where we interlock the chains by inserting $(a_i,b_{i+1})$ and $(b_i,a_{i+1})$ for two parallel edges $(a_i,a_{i+1})$ and $(b_i,b_{i+1})$.
The source \source is connected to the two start nodes of the anchor chain, the two nodes at the end of the chain are connected to \sink.
\begin{figure}
	\centering
	\begin{tikzpicture}[scale=0.55]
		\tikzstyle{new style 0}=[fill=white, draw=black, shape=circle]
		\tikzstyle{new style 1}=[fill={rgb,255: red,128; green,128; blue,128}, draw=black, shape=circle]
		
		\tikzstyle{new edge style 0}=[->]
		\tikzstyle{new edge style 1}=[-, color=orange]
		
		\node [style=new style 0] (0) at (-7, 0) {};
		\node [style=new style 0] (1) at (-6, 1) {};
		\node [style=new style 0] (2) at (-6, -1) {};
		\node [style=new style 0] (3) at (-4, 1) {};
		\node [style=new style 0] (4) at (-4, -1) {};
		\node [style=new style 0] (5) at (-2, 1) {};
		\node [style=new style 0] (6) at (-2, -1) {};
		\node [style=new style 0] (7) at (2, 1) {};
		\node [style=new style 0] (8) at (2, -1) {};
		\node [style=new style 0] (9) at (4, 1) {};
		\node [style=new style 0] (10) at (4, -1) {};
		\node [style=new style 0] (11) at (6, 1) {};
		\node [style=new style 0] (12) at (6, -1) {};
		\node [style=new style 0] (13) at (7, 0) {};
		\node  (14) at (-6, -1.75) {$a_1$};
		\node  (15) at (-4, -1.75) {$a_2$};
		\node  (16) at (-2, -1.75) {$a_3$};
		\node  (17) at (-1, 1) {};
		\node  (18) at (-1, -1) {};
		\node  (19) at (-1, 0) {};
		\node  (20) at (1, 1) {};
		\node  (21) at (1, 0) {};
		\node  (22) at (1, -1) {};
		\node  (23) at (0, 0) {...};
		\node  (24) at (-7.75, 0) {$\source$};
		\node  (25) at (7.75, 0) {$\sink$};
		\node  (26) at (-6, 1.75) {$b_1$};
		\node  (27) at (-4, 1.75) {$b_2$};
		\node  (28) at (-2, 1.75) {$b_3$};
		\draw [style=new edge style 0] (0) to (1);
		\draw [style=new edge style 0] (0) to (2);
		\draw [style=new edge style 0] (2) to (4);
		\draw [style=new edge style 0] (2) to (3);
		\draw [style=new edge style 0] (1) to (4);
		\draw [style=new edge style 0] (1) to (3);
		\draw [style=new edge style 0] (3) to (5);
		\draw [style=new edge style 0] (4) to (6);
		\draw [style=new edge style 0] (3) to (6);
		\draw [style=new edge style 0] (4) to (5);
		\draw [style=new edge style 0] (7) to (9);
		\draw [style=new edge style 0] (8) to (10);
		\draw [style=new edge style 0] (8) to (9);
		\draw [style=new edge style 0] (7) to (10);
		\draw [style=new edge style 0] (10) to (12);
		\draw [style=new edge style 0] (12) to (13);
		\draw [style=new edge style 0] (11) to (13);
		\draw [style=new edge style 0] (9) to (11);
		\draw [style=new edge style 0] (10) to (11);
		\draw [style=new edge style 0] (9) to (12);
		\draw (5) to (17.center);
		\draw (5) to (19.center);
		\draw (6) to (19.center);
		\draw (6) to (18.center);
		\draw [style=new edge style 0] (22.center) to (8);
		\draw [style=new edge style 0] (21.center) to (8);
		\draw [style=new edge style 0] (21.center) to (7);
		\draw [style=new edge style 0] (20.center) to (7);
	\end{tikzpicture}
	\label{fig:anchorChain}
	\caption{Schematic representation of an anchor chain.}
\end{figure}
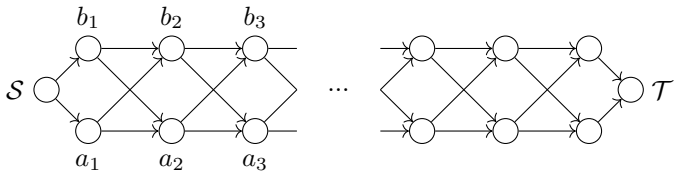

\begin{lemma}
	\label{lem:anchorChain}
	If the task graph of a \modelnameunit problem contains an anchor chain, every valid schedule has to schedule all but one of $a_1$,$b_1$ and one of $a_l$,$b_l$ on the cloud. For every job $a_i,b_i$ $1<i<l$ the time step in which it will finish processing on the cloud in every valid schedule is $i+1$.
\end{lemma}

Finally we give the reduction function $f(\phi) = G,\deadline,\budget$, where $G=(\jobs,E)$.
Set $\deadline = 12 + m + n$ and $\budget = \mid \jobs\mid  - (2 + m +n)$.
We define $G$ by constructively giving which jobs and edges are created by $f$.
Create an anchor chain of length $d-2$, this will be used to limit parts of a schedule to certain time frames.
Note that by \Cref{lem:anchorChain} we know that every valid schedule of $G=(\jobs,E),\deadline,\textbf{k}$ has every node pair of the anchor chain (besides the first and last) on the cloud at a specific fixed timestamp.
More specifically, the completion time of $a_{i}$ and $a_{i+j}$ differ by exactly $j$ time units.
For each variable $x_i \in \mathcal{X}$ create two jobs $j_{x_i}$ and $j_{\bar{x}_i}$ and edges $(a_{1+i},j_{x_i}), (a_{1+i},j_{\bar{x}_i})$ and $(j_{x_i},a_{5+i}),(j_{\bar{x}_i},a_{5+i})$.
For each clause $\clause{p}$ create a clause job $j_{\clause{p}}$ and edges $(a_{7+m+p},j_{\clause{p}})$ and $(j_{\clause{p}},a_{9+m+p})$.
Let $L_1^p, L_2^p ,L_3^p$ be the literals in $\clause{p}$.
Create jobs $j_{L_1^p},j_{L_2^p},j_{L_3^p}$ and edges $(j_{L_1^p},\clause{p}),(j_{L_2^p},\clause{p}),(j_{L_3^p},\clause{p})$ for these literals.
For every literal job $j_{L_1^p}$ connect it to the corresponding variable job $j_{x_i}$ or $j_{\bar{x}_i}$ by a chain of length $1 + ( m - i) + p$.
Also create an edge from $a_{3+i}$ to the start of the created chain and an edge from the end of the chain to $a_{6+m+p}$.
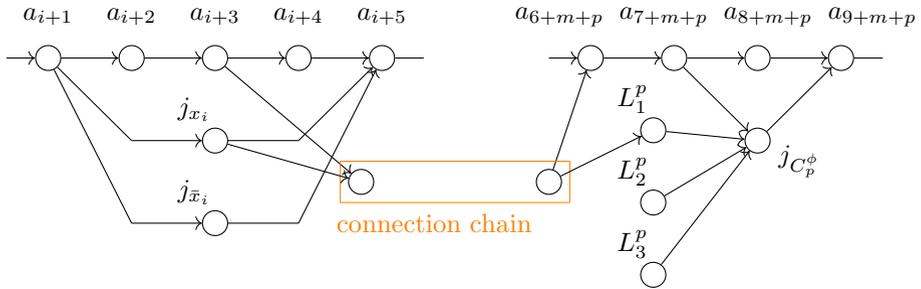
\begin{figure}
	\centering
	\begin{tikzpicture}[scale=0.55]
		\tikzstyle{new style 0}=[fill=white, draw=black, shape=circle]
		\tikzstyle{new style 1}=[fill={rgb,255: red,128; green,128; blue,128}, draw=black, shape=circle]
		
		\tikzstyle{new edge style 0}=[->]
		\tikzstyle{new edge style 1}=[-, color=orange]
		
		\node [style=new style 0] (0) at (-10, 0) {};
		\node [style=new style 0] (1) at (-8, 0) {};
		\node [style=new style 0] (2) at (-6, 0) {};
		\node  (3) at (-10, 1) {$a_{i+1}$};
		\node  (4) at (-8, 1) {$a_{i+2}$};
		\node  (5) at (-6, 1) {$a_{i+3}$};
		\node [style=new style 0] (6) at (-4, 0) {};
		\node [style=new style 0] (7) at (-2, 0) {};
		\node  (9) at (-4, 1) {$a_{i+4}$};
		\node  (10) at (-2, 1) {$a_{i+5}$};
		\node [style=new style 0] (11) at (-6, -2) {};
		\node [style=new style 0] (12) at (-6, -4) {};
		\node  (13) at (-11, 0) {};
		\node  (14) at (-1, 0) {};
		\node  (15) at (-8, -2) {};
		\node  (16) at (-8, -4) {};
		\node  (17) at (-4, -2) {};
		\node  (18) at (-4, -4) {};
		\node [style=new style 0] (19) at (3, 0) {};
		\node [style=new style 0] (20) at (5, 0) {};
		\node [style=new style 0] (21) at (7, 0) {};
		\node [style=new style 0] (22) at (9, 0) {};
		\node [style=new style 0] (23) at (7, -2) {};
		\node [style=new style 0] (24) at (4.5, -1.75) {};
		\node [style=new style 0] (25) at (4.5, -3.5) {};
		\node [style=new style 0] (26) at (2, -3) {};
		\node [style=new style 0] (27) at (4.5, -5.25) {};
		\node [style=new style 0] (28) at (-2.5, -3) {};
		\node  (29) at (2.25, 1) {$a_{6+m+p}$};
		\node  (30) at (4.75, 1) {$a_{7+m+p}$};
		\node  (31) at (7.25, 1) {$a_{8+m+p}$};
		\node  (32) at (9.75, 1) {$a_{9+m+p}$};
		\node  (33) at (-3, -2.5) {};
		\node  (34) at (-3, -3.5) {};
		\node  (35) at (2.5, -2.5) {};
		\node  (36) at (2.5, -3.5) {};
		\node [color=orange] (37) at (-0.75, -4) {connection chain};
		\node  (38) at (2, 0) {};
		\node  (39) at (10, 0) {};
		\node  (40) at (8, -2.5) {$j_{\clause{p}}$};
		\node  (41) at (-6.5, -1.25) {$j_{x_i}$};
		\node  (42) at (-6.5, -3.25) {$j_{\bar{x}_i}$};
		\node  (43) at (4, -1) {${L_1^p}$};
		\node  (44) at (4, -2.75) {${L_2^p}$};
		\node  (45) at (4, -4.5) {${L_3^p}$};
		\draw [style=new edge style 0] (0) to (1);
		\draw [style=new edge style 0] (1) to (2);
		\draw [style=new edge style 0] (6) to (7);
		\draw [style=new edge style 0] (2) to (6);
		\draw [style=new edge style 0] (13.center) to (0);
		\draw (7) to (14.center);
		\draw (0) to (15.center);
		\draw (0) to (16.center);
		\draw (11) to (17.center);
		\draw (12) to (18.center);
		\draw [style=new edge style 0] (17.center) to (7);
		\draw [style=new edge style 0] (18.center) to (7);
		\draw [style=new edge style 0] (15.center) to (11);
		\draw [style=new edge style 0] (16.center) to (12);
		\draw [style=new edge style 0] (19) to (20);
		\draw [style=new edge style 0] (20) to (21);
		\draw [style=new edge style 0] (21) to (22);
		\draw [style=new edge style 0] (20) to (23);
		\draw [style=new edge style 0] (23) to (22);
		\draw [style=new edge style 0] (24) to (23);
		\draw [style=new edge style 0] (25) to (23);
		\draw [style=new edge style 0] (27) to (23);
		\draw [style=new edge style 0] (26) to (19);
		\draw [style=new edge style 0] (2) to (28);
		\draw [style=new edge style 1] (33.center) to (34.center);
		\draw [style=new edge style 1] (34.center) to (36.center);
		\draw [style=new edge style 1] (36.center) to (35.center);
		\draw [style=new edge style 1] (35.center) to (33.center);
		\draw [style=new edge style 0] (38.center) to (19);
		\draw [style=new edge style 0] (11) to (28);
		\draw [style=new edge style 0] (26) to (24);
		\draw (22) to (39.center);
	\end{tikzpicture}
	\caption{Schematic representation of the variable and clause gadgets and their connection.}
	\label{fig:variableGadget}
\end{figure}

It remains to show that there is a schedule of length at most $\deadline$ with costs at most $\budget$ in $f(\phi) = G,\deadline,\budget$ if and only if there is a satisfying assignment for $\phi$.
\begin{lemma}
	\label{lem:fixedOnCloud}
	In a deadline adhering schedule for $f(\phi) = G,\deadline,\budget$ every job in the anchor chain (except on at the front and one at the end), every job in the variable and clause literal connecting chains and every clause job has to be scheduled on the cloud.
\end{lemma}
\begin{proof}
	By \Cref{lem:anchorChain} we already know that every node in the anchor chain except one of $v_1$,$w_1$ and one of $v_l$,$w_l$ has to be scheduled on the cloud.
	We also know, that the jobs in the anchor chain have fixed time steps in which they have to be processed.
	We look at some chain and its connection to the anchor chain. 
	The start of the chain of length $1 + ( m - i) + p$ is connected to $a_{3+i}$, the end to $a_{6+m+p}$. Between the end of $a_{3+i}$ and the start of $a_{6+m+p}$ are $6+m+p - 1 - (3+i) = 2 +m + p -i$ time steps.
	So with the processing time required to schedule all $1 + ( m - i) + p$ jobs of the chain, there is only one free time step, but we would need at least 2 free time steps to cover the communication cost to and from the server. 
	(Recall here that both $a_{3+i}$ and $a_{6+m+p}$ have to be processed on the cloud).
	The same simple argument fixes each clause job to a specific time step on the server. 
\end{proof}

\begin{lemma}
	\label{lem:onePerVariableClause}
	In a deadline adhering schedule for $f(\phi) = G,\deadline,\budget$ only one of $j_{x_i}$ and $j_{\bar{x}_i}$ can be processed on the server for every variable $x_i \in \mathcal{X}$.
	The same is true for $j_{L_1^p},j_{L_2^p},j_{L_3^p}$ of clause $\clause{p}$.
\end{lemma}
\begin{proof}
	$j_{x_i}$ and $j_{\bar{x}_i}$ are both fixed to the same time interval via the edges $(a_{1+i},j_{x_i})$, $ (a_{1+i},j_{\bar{x}_i})$ and $(j_{x_i},a_{5+i}),(j_{\bar{x}_i},a_{5+i})$.
	Since $a_{1+i}$ and $a_{5+i}$ will be processed on the cloud and keeping communication delays in mind, only the middle of the three time steps in between can be used to schedule $j_{x_i}$ or $j_{\bar{x}_i}$ on the server.
	Since the server is only a single machine only on of them can be processed on the server.
	Note here that the other job can be scheduled a time step earlier which we will later use.
	The argument for $j_{L_1^p},j_{L_2^p},j_{L_3^p}$ works analogously to the statement above.
	
\end{proof}

\begin{lemma}
	\label{lem:scheduleIsAssignment}
	There is a deadline adhering schedule for $f(\phi) = G,\deadline,\budget$ with costs of $\mid \jobs\mid  - (2 + m +n)$ if and only if there is a satisfying assignment for $\phi$.
	The variable jobs processed on the cloud represent this satisfying assignment
\end{lemma}
\begin{proof}
	From \Cref{lem:anchorChain}, \Cref{lem:fixedOnCloud} and \Cref{lem:onePerVariableClause} we can infer that a schedule with costs of $\mid \jobs\mid  - (2 + m +n)$ has two jobs of the anchor chain, one job for each pair of variable jobs and one job per clause on the server.
	Two jobs of the anchor chain can always be placed on the server, the choice of variable jobs is also free.
	It remains to show, that we can only schedule a literal job per clause on the server if and only if the respective clause is fulfilled by the assignment inferred by the variable jobs.
	
	The clause job $j_{\clause{p}}$ of $\clause{p}$ has to be processed in time step $9+m+p$ (between $a_{7+m+p}$ and $a_{9+m+p}$).
	Therefore, $j_{L_1^p}$ has to be processed no later than $8+m+p$ or $7+m+p$ if it is processed on the cloud or server respectively.
	Let $j_{x_i}$ be the variable job connected to $j_{L_1^p}$ via a connection chain.
	
	If $j_{x_i}$ is \emph{true} (scheduled on the cloud), it can finish processing at time step $3+i$, which does not delay the start of the connection chain (which is connected to $a_{3+i}$, finishing in time step $4+i$).
	This means that the chain can finish in time step $4+i~+~1 + ( m - i) + p~=~5+m+p$, the time step $6+m+p$ can be used for communication, allowing $j_{L_1^p}$ to be processed by the server in $7+m+p$.
	
	If $j_{x_i}$ is \emph{false} (scheduled on the server), it finishes processing at time step $4+i$, which, combined with the induced communication delay, delays the start of the chain by $1$.
	Therefore, the chain only finishes in time step $6+m+p$, and $j_{L_1^p}$ has to be processed on the cloud, since there is not enough time for the communication back and forth.
	
	Trivially, the same argument holds true for $j_{L_2^p}$ and $j_{L_3^p}$.
	
\end{proof}

It should be easy to see that the reduction function $f$ is computable in polynomial time.
Combined with \Cref{lem:scheduleIsAssignment} this concludes the proof of our reduction $3SAT \leq_p \modelnameunit$.
The correctness of \Cref{the:unitStrongNP} trivially follows from that.

\subsubsection{The General Case}
Adapting the previous reduction we can show an even stronger result for the general case of \modelname.
Basically we are able to degenerate the reduction output in a way, that a satisfying assignment results in a schedule with cost $0$, while every other assignment (schedule) has costs of at least $1$.
It should be obvious, that this also means that there is no approximation algorithm for this problem with a fixed multiplicative performance guarantee, if $\text{P} \neq \text{NP}$.

This reduction uses processing times and communication delays of $0$, $\infty$ and values in between.
Note that $\infty$ can simply be replaced by $\deadline+1$.
To keep the following part readable we again substitute \enquote{an edge $(j,j')$ with communication delay $c(j,j')=k$} simply by \enquote{an edge $c(j,j')=k$}

We follow the same general structure (an anchor chain, variable-, clause- and connection gadgets).
The \emph{anchor chain} now looks as follows: For every time step create two jobs $a_i$ and $a_i'$ with $\procSF{a_i}=0$, $\procCF{a_i}=\infty$, $\procSF{a_i'}=\infty$, $\procCF{a_i'}=0$ and an edge $c(a_i,a_i')=0$.
These chain links are than connected by an edge $c(a_i',a_{i+1})=1$.
Finally we create $c(\source,a_1)=1$ and $c(a_\deadline,\sink)=0$.
It should be easy to see, that every schedule will process $a_i$ and $a_i'$ in time step $i$ on the server and the cloud respectively.
This gives us anchors to the server and to the cloud for every time step, without inducing congestion or costs.
Since the anchor jobs themselves have processing time of $0$, the \enquote{usable} time interval between some $a_i$ and $a_{i+1}$ is one full time step.

For each variable $x_i \in \mathcal{X}$ create two jobs $j_{x_i}$, $j_{\bar{x}_i}$ with $\procSF{j_{x_i}} = \procSF{j_{\bar{x}_i}} = 1$ and $\procCF{j_{x_i}} = \procCF{j_{\bar{x}_i}} = 0$.
Create edges $c(a_{i},j_{x_i})=1$, $c(a_{i},j_{\bar{x}_i})=1$ and $c(j_{x_i},a_{i+1})=0$, $c(j_{\bar{x}_i},a_{i+1})=0$.
In short, only one of them can be processed on the server, the other on the cloud.
Both will finish in time step $i+1$, the one processed on the server is \emph{true}, therefore processing both on the cloud is possible, but not helpful.

For each clause $\clause{p}$ create a clause job $j_{\clause{p}}$ with $\procSF{j_{\clause{p}}}=\infty$, $\procCF{j_{\clause{p}}}=0$ and edges $c(a_{5+m+3p}',j_{\clause{p}})=\infty$ and $c(j_{\clause{p}},a_{6+m+3p}')=\infty$.
This means, that $j_{\clause{p}}$ has to finish processing by time step $6+m+3p$.
Let $L_1^p, L_2^p ,L_3^p$ be the literals in $\clause{p}$.
Create jobs $j_{L_1^p},j_{L_2^p},j_{L_3^p}$ each with $\procC=\procS=1$ and edges $c(j_{L_1^p},\clause{p})=0$, $c(j_{L_2^p},\clause{p})=0$, $c(j_{L_3^p},\clause{p})=0$ for these literals.
Create edges $c(a_3+m+3p,j_{L_1^p})=0$, $c(a_3+m+3p,j_{L_2^p})=0$ and $c(a_3+m+3p,j_{L_3^p})=0$, so that, in theory, all three of the literal jobs can be processed on the server, finishing in time steps $4+m+3p$, $5+m+3p$ and $6+m+3p$ respectively.
Lastly, for every literal job $j_{L_1^p}$ connect it to the corresponding variable job $j_{x_i}$ (or $j_{\bar{x}_i}$) by a an edge with communication delay of $m-i+3p+3$.
Since $j_{x_i}$ (or $j_{\bar{x}_i}$) finish processing in time step $i+1$, this means that $j_{L_1^p}$ can start no earlier than $m + 3p + 4$ (and therefore finish processing in $5+m+3p$), if $j_{x_i}$ (or $j_{\bar{x}_i}$) were processed on the cloud.

Recall here, that a variable job being scheduled on the server denotes that it is \emph{true}.
So only a literal job that evaluates to true, can be scheduled so that it finishes processing in time step $4+m+3p$ on the cloud.

It follows directly, that a schedule for this construction will have costs of $0$ if and only if the assignment derived from the placement of the variable jobs fulfills every clause.

\begin{theorem}
	There is no approximation algorithm for \modelname that has a fixed performance guarantee, assuming that $P\neq NP$.
\end{theorem}

\section{Unit Size and Unit Delay - And no Delay}\label{sec:unit_size_unit_delay}
As the last step of this paper we explore simple algorithms on unit size instances with arbitrary task graphs.
Recall that we proved these to be strongly NP-hard.
We use resource augmentation and ask: given a \modelnameunit problem instance with deadline \deadline, find a schedule in poly. time that has a makespan of at most $(1+\varepsilon)\cdot \deadline$ that \emph{approximates} the optimal cost in regards to the actual deadline~\deadline.

If there is a chain of length $\deadline$ or $\deadline-1$, that chain has to be scheduled on the server, since there is no time for the communication delay.
For instances with a chain of size $\deadline$ that is trivially optimal, for those with $\deadline-1$ we can check in polynomial time if any other job also fits on the server, again, finding an optimal solution.
From now we assume that there is no chain of length more than $\deadline-2$.

First, construct a schedule which places every job on the cloud, as fast as possible.
The resulting schedule from time step (ts) 1 to $(1+\varepsilon)\cdot \deadline$ looks as follows: one ts of communication, at most $\deadline-2$ ts of processing on the server, another ts for communication followed by at least $\varepsilon\deadline$ empty ts.
Now pull (one of) the last job(s) that is processed on the cloud to the last empty ts and process it on the server instead.
Repeat this process until the last job can not be moved to the server anymore.
Do the whole procedure again, but this time starting with the cloud schedule in the end of the schedule, and each time pulling the first job to the beginning.
Keep the result with lower costs.
Note that one can always fill the ts being used solely for communicating from the server to the cloud with processing one job on the server, that otherwise would be one of the first jobs being processed on the cloud (the same holds for the other direction).

\begin{theorem}
	The described algorithm yields a schedule with approximation factor of $\frac{1+\varepsilon}{2\varepsilon}$ while having a makespan of at most $(1+\varepsilon)\cdot \deadline$.
\end{theorem}
\begin{proof}
	Case $n \leq (1+\varepsilon ) \deadline$: 
	The algorithm places all jobs on the server, the cost is $0$ and therefore optimal.
	
	\noindent
	Case $(1+\varepsilon ) \deadline < n < (1+2\varepsilon ) \deadline$:
	Assume that the preliminary cloud-only schedule needs $\deadline-2$ ts on the cloud, if that is not the case, we \emph{stretch} the schedule to that length.
	There are $n$ jobs distributed onto $\deadline-2$ ts.
	Therefore, either from the front or from the end, there is an interval of length $\frac{\deadline}{2}-1$ with at least $\frac{\deadline}{2}-1$ and at most $\frac{n}{2} < \frac{(1+2\varepsilon ) \deadline}{2} = \frac{\deadline}{2}+\varepsilon\deadline$ many jobs.
	It should be easy to see, that the algorithm will schedule those at most $\frac{\deadline}{2}+\varepsilon\deadline-1$ jobs to the $\frac{\deadline}{2}-1$ plus the free $\varepsilon\deadline$ many time slots.
	If the interval included less than $\frac{\deadline}{2}+\varepsilon\deadline-1$ jobs, it will simply continue until the $\frac{\deadline}{2}-1 + \varepsilon\deadline$ ts are filled with jobs being processed on the server.
	With the one job we can process on the server during the communication ts we process $\frac{\deadline}{2} + \varepsilon\deadline$ jobs on the server and have costs of $n - (\frac{\deadline}{2} + \varepsilon\deadline)$.
	An optimal solution has costs of at least $n - \deadline$.
	For $\varepsilon \geq 0.5$ it holds that: $\costAlg = n - (\frac{\deadline}{2} + \varepsilon\deadline) \leq n-d \leq \costOpt$, otherwise:
	\[ \frac{\costAlg}{\costOpt} \leq \frac{n -( \frac{\deadline}{2} + \varepsilon\deadline)}{n - \deadline} \leq \frac{(1+\varepsilon ) \deadline - (\frac{\deadline}{2} + \varepsilon\deadline)}{(1+\varepsilon ) \deadline - \deadline} \leq \frac{0.5\deadline}{\varepsilon\deadline} = \frac{1}{2\varepsilon} \]
	
	\noindent
	Case $(1+2\varepsilon ) \deadline \leq n$:
	In this case we simply observe that our algorithm places at least $\varepsilon\deadline$ many jobs on the server.
	For $\varepsilon \geq 1$ it holds that: $\costAlg = n - \varepsilon\deadline \leq n-d \leq \costOpt$, otherwise:
	\[ \frac{\costAlg}{\costOpt} \leq \frac{n - \varepsilon\deadline}{n - \deadline} \leq \frac{(1+2\varepsilon ) \deadline  - \varepsilon\deadline}{(1+2\varepsilon ) \deadline  - \deadline} = \frac{\deadline+\varepsilon\deadline}{2\varepsilon\deadline} = \frac{1+\varepsilon}{2\varepsilon} \]
\end{proof}

\subsection{No Delays and Identical Machines}
We design a simple heuristic for the case in which the server and the cloud machines behave the same, that is, $\procC(j) = \procS(j)$ for each job $j$ (except for the source and sink), and the communication delays all equal zero.
In this case, we may define the length of a chain in the task graph as the sum of the processing times of the jobs in the chain.
The first step in the algorithm is to identify a longest chain in the task graph, which can be done in polynomial time.
The jobs of the longest chain are scheduled on the server and the remaining jobs on the cloud each as early as possible.
Now, the makespan of the resulting schedule is the length of a longest chain, which is optimal (or better) and there are no idle times on the server.
However, the schedule may not be feasible since the budget may be exceeded.
Hence, we repeatedly do the following:
If the budget is still exceeded, we pick a job scheduled on the cloud with maximal starting time and move it on to the server right before its first successor (which may be the sink).
Some jobs on the server may be delayed by this but we can do so without causing idle times.
If all the processing times are equal this procedure produces an optimal solution and otherwise there may be an additive error of up to the maximal job size.
Hence, we have:
\begin{theorem}
	There is a $2$-approximation for \modelname without communication delays and identical server and cloud machines.
\end{theorem}
It is easy to see, that the analysis is tight considering an instance with three jobs: 
One with size $\budget$, one with size $\budget+\varepsilon$, and one with size $2\varepsilon$.
The first jobs precedes the last one.
Our algorithm will place everything on the server, while the first job is placed on the cloud in the optimal solution.

Note that we can take a similar approach to find a solution with respect to the cost objective by placing more and more jobs on the server as long as the deadline is still adhered to.
However, an error of one job can result in an unbounded multiplicative error in the objective in this case.
On the other hand, it is easy to see that in the case with unit processing times, there will be no error at all in both procedures yielding:
\begin{corollary}
	The variant of \modelname without communication delays and unit processing times can be solved in polynomial time with respect to both the makespan and the cost objective. 
\end{corollary}

\section{Generalizations of Server Cloud Scheduling} \label{sec:generalization}
In this chapter we introduce some generalizations to the \modelname. 
We consider different aspects from multiple clouds and server machines to direction specific delays.
We sketch how to adapt our algorithms for \modelnameextended and \modelnameconst to cover those new generalizations.

\subsection{Changes in the Definitions}
We shortly define the changes to the model that we explore in this section.

\subsubsection{Machine Model}
So far we imagined a single server machine and one homogeneous cloud in our problem definition. 
Now, instead of a single server machine there can be any (constant) number of identical server machines: $\textsc{server}=\{s_1,\dots, s_z\}$.
Instead of one homogeneous cloud there can be any number of different cloud contexts: $\textsc{clouds}=\{c_1,\dots, c_k\}$.
Each cloud context still consists of an unlimited number of parallel machines.

\subsubsection{Jobs}
Jobs are still given as a task graph $G=(\jobs, E)$.
A job $j \in \jobs$ has processing time \procSF{j} on any server machine and processing time \procCFmc{j}{i} on a machine of cloud context $c_i$.
An edge $e = (i,j)$ and machine contexts $m_1, m_2 \in \{s, c_1,\dots, c_k\}$ have a communication delay of $\comFmc{{m_1}}{{m_2}}{i,j} \in \mathbb{N}_0$, which means, that after job $i$ finished on a machine of type $m_1$, $j$ has to wait an additional $\comFmc{{m_1}}{{m_2}}{i,j}$ time steps before it can start on a machine of type $m_2$.
For $m_1 = m_2$ we set $\comFmc{{m_1}}{{m_2}}{i,j} = 0$.
Note that this function does not need to be symmetric, e.g. $\comFmc{{m_1}}{{m_2}}{i,j}$ and $\comFmc{{m_2}}{{m_1}}{i,j}$ may be unequal.

\subsubsection{Costs and Schedules}
Previously we defined cost simply by \enquote{time spend on the cloud}.
While considering multiple clouds, that is not sensible anymore.
A faster cloud will not be universally cheaper than a slower one.
We define a cost function based on the cloud context and job, $cost: \jobs \times \textsc{clouds} \mapsto \mathbb{N}_0$.
A schedule still consists of $C: \jobs \mapsto \mathbb{N}_0$ (maps jobs to their completion time), but instead of a partition we give a mapping function $\eta: \jobs \mapsto \{s_1,\dots, s_z\} \cup \{c_1,\dots, c_k\}$.
Note that $s_i$ refers to one specific server machine, while $c_i$ refers to a cloud context, consisting of infinitely many machines.

We call a schedule $\pi = (C, \eta)$ valid if and only if the following conditions are met:
\begin{enumerate}
	\item[a)] There is always at most one job processing on each server:\\
	\[{\forall}_{i, j \in \jobs, i\neq j: \eta(i)=\eta(j)\in \textsc{server}}:
	(\completion{i} \leq \completion{j}-\procSF{j}) \vee (\completion{i}-\procSF{i} \geq \completion{j}) \]
	\item[b)] Tasks are not started before the previous tasks has been finished/ the required communication is done:\\
	\[\forall_{(i,j) \in E}: (\completion{i}+\comFmc{\eta(i)}{\eta(j)}{i,j} \leq \completion{j}- p_{\eta(j)}{j})\]
\end{enumerate}
The makespan ($\makespan$) of a schedule is still given by the completion time of the sink \sink: \completion{\sink}.
The cost ($\cost$) of a schedule is given by:
\[\sum_{j\in jobs: \eta(j)\in\textit{clouds}} cost(j,\eta(j)).\]

\subsection{Revisiting \modelnameextended}
We briefly sketch how to adapt the algorithm from \Cref{sec:extendedchain} to incorporate the previously defined changes on the model.
We will use the observations, that multiple server machines only affect the scheduling of parallel parts and that we can always calculate an optimal cloud location for a job in a given situation (part of the schedule, time frame and location of predecessor and successor).

\begin{theorem}
	There is a $(4+\varepsilon)$-approximation algorithm for the budget restrained makespan minimization problem on extended chains, even when there are $z$ server machines, $k$ different cloud contexts, the communication delays are directionally dependent on the machine context, and costs are given as an arbitrary cost function $cost: \jobs \times \textsc{clouds} \mapsto \mathbb{N}_0$.
\end{theorem}

\begin{proof}
	We adapt the pseudo polynomial algorithm from \Cref{sec:extendedchain} that given a feasible makespan estimate $T (T\geq \makeOpt)$ calculates a schedule with makespan of at most $\min \{2\makeEstimate, 2\makeOpt \}$, such that it incorporates the changes to the model and calculates a schedule with makespan of at most $\min \{4\makeEstimate + \varepsilon', 4\makeOpt + \varepsilon'\}$.
	The only change in the state description is that $loc \in \{s,c_1,\dots,c_k\}$ instead of $loc \in \{s,c\}$.
	As the state description is used for the chain parts of the extended chain, we do not differentiate the server machines here.
	The creation of the state extension list $\textsc{Extensions}^{j}$ (each of form $[\Delta t,\dynLocE{j-1}\rightarrow \dynLocE{j}] = \cost$), has the following changes:
	\begin{itemize}
		\item Instead of the four combinations $s \rightarrow s$, $s \rightarrow c$, $c \rightarrow s$, $c \rightarrow c$, we consider all combinations from $\{s,c_1,\dots,c_k\}\times \{s,c_1,\dots,c_k\}$.
		\item Substitute the corresponding values, for example $[\procCF{j} + \comF{j-1,j},~s\rightarrow c] = \procCF{j}$ becomes $[p_{c_i}(j) + \comFmc{s}{{c_i}}{j-1,j},~s\rightarrow c_i] = cost(j,c_i)$.
		\item If there is a parallel subgraph between $j-1$ and $j$ we adapt the calculation in the following way:
		\begin{itemize}
			\item Calculate $\Delta^{max}$ as before (the sum over all processing times on the server plus the biggest relevant in- and outgoing communication delays)
			\item Iterate over $\Delta^i$ in $\{0,\dots, \Delta^{max} \}$:
			\begin{itemize}
				\item As before, check for each job if it fits: (1) only on the servers, (2) not on the servers but on at least one cloud context, (3) on both, (4) on none. If at least one job falls into (4) break.
				\item Calculate for each job $j$ in (2) or (3) the cheapest fitting option to schedule that job on some available cloud in time frame $\Delta^i$. Use that cost $c_j$ for $j$ for the remainder of the iteration.
				\item Greedily put jobs in (1) onto server machines ($1$ to $k$) until the current server has load $\geq \Delta^i$, proceed with the next machine and so on. If not all jobs in (1) can be placed this way break, as there is not enough space to place jobs on the server that do not fit on the cloud in the given time frame.
				\item Sort the jobs in (3) by their ratio of cost $c_j$ to processing time on the server (highest to lowest cost per time). Continue by greedily placing those on the server machines as before. When all jobs in (3) are placed, or all server machines have load $\geq \Delta^i$, put all remaing jobs from (3) on their corresponding cheapest cloud context.
				\item Put all jobs from (2) on their corresponding cheapest cloud context.		
				\item insert time in the front and back corresponding to the biggest communication delay invoked by the (sub-)schedule for the parallel part
			\end{itemize}
		\end{itemize}
	\end{itemize} 
	The rest of the algorithm behaves as before.
	The changes to state extensions spanning a parallel subgraph calculate solutions that have at most optimal cost for a time frame of $\Delta^i$, while using a time frame of $4\Delta^i$.
	The $4$ times correspond to: at most $2\Delta^i$ time for all in- and outgoing communication delays since the communication delays have to fit into $\Delta^i$ to be considered, at most $2\Delta^i$ time for our greedy packing of the server machines since we can add a job of size $\Delta^i$ to a machine currently having load $\Delta^i - \epsilon$.
	It should be easy to see that the greedy packing of \enquote{highest cost jobs}, with what is essentially resource augmentation of a multiple knapsack problem, gives at most optimal cost.
	Note that we could also utilize a PTAS for multiple knapsack here to stay in a time frame of $3\Delta^i$, but we want to find a solution with optimal cost (or lower), to remain strictly budget adhering.
	
	It remains to simply use the same scaling technique used in \Cref{sec:extendedchain} to get the $4+\varepsilon$-approximation.
	
\end{proof}
If the communication delays are constant the result can be easily adapted to yield a $2+\varepsilon$-approximation, by getting rid of the added time for communication delays.

\subsection{Revisiting \modelnameconst}
In a similar vein as the previous subsection we briefly sketch how to adapt the results from \Cref{sec:constant_width_FPTAS} to include most of the previously defined model generalizations.
Naturally, we still require the \emph{maximum cardinality source and sink dividing cut} to be bounded by a constant.
In contrast to the previous result we require the number of server machines to be a constant.

\begin{theorem}
	There is an FPTAS for the budget restrained makespan minimization problem for graphs with a constant \emph{maximum cardinality source and sink dividing cut}, even when there are a constant number of server machines, $k$ different cloud contexts, the communication delays are directionally dependent on the machine context, and costs are given as an arbitrary cost function $cost: \jobs \times \textsc{clouds} \mapsto \mathbb{N}_0$.
\end{theorem}

\begin{proof}
	We make the following two changes to the state definition:
	We consider $\dynLocE{j}\in \{s,c_1,\dots,c_k\}$ instead of $\dynLocE{j}\in \{s,c\}$, we track the unused time of every server machine individually so instead of a single $f_s$ the state contains $f_{s^1},\dots,f_{s^z}$.
	The dynamic program needs only minor tweaks.
	When iterating through the jobs that are open (and of which all predecessors have been processed) use the server $s^i$ with the smallest fitting $f_{s^i}$ and set $f_{s^i}=0$.
	Instead of checking if the job fits on \enquote{the cloud} we simply go through all clouds, and add corresponding states for each fitting location.
	While calculating the value of a state use the new cost function $cost$ instead of $p_c$, while checking if a job fits we use the directional communication delays.
	After a full iteration increase each $f_{s^i}$ by one (instead of only increasing the singular $f_s$).
	It should be easy to see, that these adaptations do not change the correctness of the algorithm.
	The runtime (after the rounding technique) naturally increases to $poly(n^z,k,\frac{1}{\varepsilon})$, which is polynomial, iff $z$ (the number of server machines) is a constant.
	
\end{proof}

\section{Approximating the Pareto Front}\label{sec:pareto}

The problem variants we describe and analyze in this paper are multi-criteria optimization problems.
To simultaneously handle the two criteria cost and makespan, we either looked at decision variants \enquote{is there a schedule with makespan $\leq \deadline$ and cost $\leq \budget$} or we used one of them as a constraint and asked \enquote{given a budget of $\budget$, minimize the makespan} (or vice versa).
Naturally, one might be interested in finding an assortment of different efficient solutions, without giving a specific budget or deadline.
A solution is called efficient, or Pareto optimal, if we can not improve one of the criteria, without worsening the other.
The set of all Pareto optimal solutions is called the Pareto front.
In the following, we will use the term \emph{point} to refer to the makespan and cost of a feasible solution of a given \modelname problem.

For our NP-hard problems, we will not be able to efficiently calculate the exact Pareto front, but we can find a set of points that is close to the optimum.
In the literature, one can find slightly different definitions for such approximations.
In \cite{DBLP:conf/tacas/LegrielGCM10}, the authors scale each criteria to an interval from $0$ to $1$.
A set of points is an $\alpha$-approximation, if for each point in the actual Pareto front, there is a point where each dimension is offset by at most an additional $\pm\alpha$.
We follow the definition of Pareto front approximations given in \cite{DBLP:conf/focs/PapadimitriouY00} (adapted to our case with exactly 2 objectives):
\begin{definition}
	A set of points $S$ is an $\alpha$-approximation of a Pareto front, if for each point $p=(mspan^p,cost^p)$ there is a point $p'=(mspan^{p'},cost^{p'})$ in $S$ with $mspan^{p'} \leq (1+\alpha) mspan^{p}$ and $cost^{p'} \leq (1+\alpha) cost^{p}$.
\end{definition}

The dynamic programming algorithms established in this paper can be used to find such an approximation.
We use the results from \Cref{sec:constant_width_FPTAS} to show how this is done, but note that a similar approach can be used for other results of this paper.

Intuitively our dynamic programs calculate a collection of possible results but only report a single one, where the \enquote{best} is selected based on the current objective.
Imagine that one of our deadline restrained algorithms with approximation factor $(1+\varepsilon)$ reports \emph{every} non dominated solution it finds instead.
The result for $\deadline = 10$ and $\varepsilon = 0.1$ could look like \Cref{fig:paretoFront}.
For every reported point $(mspan,cost)$ we can infer a lower bound on the makespan of $mspan - \varepsilon \cdot \deadline$ any schedule with a given $cost$ has, due to the approximation factor of the algorithm.
Note that gap is in relation to a given $d$, and therefore results with a smaller makespan are less precise.
We will circumvent that by repeating the algorithm with smaller values for $d$.
\begin{figure}[H]
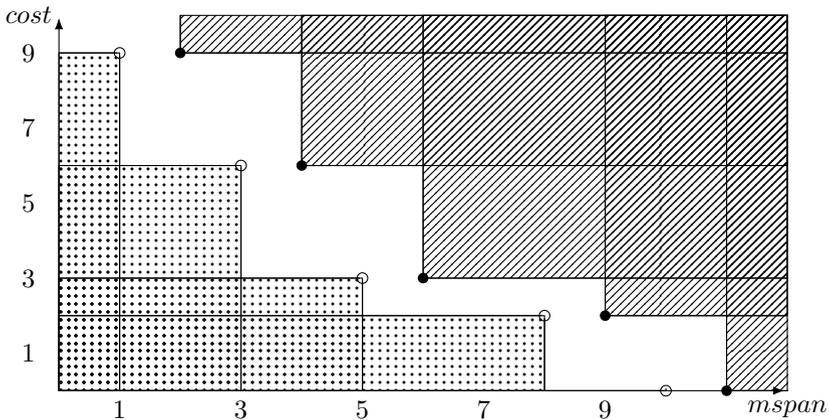

	\centering
	\tikz[>={Latex[length=1.5mm]}, shorten <= 0pt, shorten >= 1pt]{
		\pgfmathsetmacro{\w}{0.8}
		\pgfmathsetmacro{\h}{0.5}
		
		\pgfmathsetmacro{\wMax}{12}
		\pgfmathsetmacro{\hMax}{10}
		
		\newcommand*\paretoLower[2]{  
			\draw[pattern=dots] (0,0) rectangle (#1*\w,#2*\h);
			\draw[-] (0,#2*\h) -- (#1*\w,#2*\h);
			\draw[-] (#1*\w,0) -- (#1*\w,#2*\h);
			\draw (#1*\w,#2*\h) circle[radius=2pt];
		}
		
		\newcommand*\paretoUpper[2]{  
			\draw[pattern=north east lines] (\wMax*\w,\hMax*\h) rectangle (#1*\w,#2*\h);
			\draw[-] (\wMax*\w,#2*\h) -- (#1*\w,#2*\h);
			\draw[-] (#1*\w,\hMax*\h) -- (#1*\w,#2*\h);
			\fill (#1*\w,#2*\h) circle[radius=2pt];
		}

		\draw[->] (0,0) -- (0,\hMax*\h);
		\foreach \x in {1,3,...,9}
		{
			\node[] at (-0.5*\w,\h*\x) {$\x$};
		}
		\node[] at (-0.5*\w,\h*\hMax) {$cost$};
		
		\draw[->] (0,0) -- (\wMax*\w,0);
		\foreach \x in {1,3,...,9}
		{
			\node[] at (\w*\x, -0.5*\h) {$\x$};
		}
		\node[] at (\w*\wMax, -0.5*\h) {$mspan$};
		
		\paretoLower{10}{0}
		\paretoLower{8}{2}
		\paretoLower{5}{3}
		\paretoLower{3}{6}
		\paretoLower{1}{9}
		\paretoUpper{11}{0}
		\paretoUpper{9}{2}
		\paretoUpper{6}{3}
		\paretoUpper{4}{6}
		\paretoUpper{2}{9}	
	}
	\caption{Reported solutions by our algorithm, filled circles and empty circles represent reported points and best possible solutions due to the approximation factor, respectively. Dotted region is infeasible, striped region is feasible but dominated.}
	\label{fig:paretoFront}
\end{figure}

\begin{theorem}
	Using \dynProgName (\Cref{alg:dynProg}) one can $\alpha$-approximate the Pareto front of a \modelname problem with constant \dynConc in polynomial time, for any $\alpha > 0$.
\end{theorem}
\begin{proof}
	Given some \modelname problem with constant \dynConc run \dynProgName (with the rounding approach) with $\deadline = \sum_{j \in \jobs} \procSF{j}$.
	Normally the algorithm found the first state $[\scaled{\deadline}, \dynFreeServer] = cost$.
	Now, instead let the algorithm find the first state $[\dynTimestamp, \dynFreeServer] = cost$ for every $t \in (0.5\scaled{\deadline}, \scaled{\deadline}]$.
	For each of those states calculate an upper bound on the makespan for the respective schedule in the unscaled instance.
	Following the argumentation in the proof for \Cref{the:dynResult}, we know that the makespan is $\leq t+\scale + (n-2)2\scale = (t+2n-4)\scale$. 
	Report the point $(mspan= (t+2n-4)\scale, cost)$ and add it to $S$.
	After that full algorithm iteration, set $d := 0.5d$ and repeat the process.
	Do this until $d = 1$.
	Finally, return the reported point set $S$.
	
	
	
	
	We want to show that for every point $p=(mspan^p,cost^p)$ of a Pareto front, there is a reported point $p'=(mspan^{p'},cost^{p'})$ with $mspan^{p'} \leq (1+\alpha) mspan^{p}$ and $cost^{p'} \leq (1+\alpha) cost^{p}$.
	Given some point $p=(mspan^p,cost^p)$, look at the iteration where $0.5\deadline < mspan^p \leq \deadline$.
	Since there is a feasible schedule with $mspan^p$ and $cost^p$ at some point during that iteration we found a feasible scaled schedule with $t = \lfloor \frac{mspan^p}{\scale} \rfloor$ and $cost \leq cost^p$.
	The calculated upper bound for that schedule in unscaled is then $(\lfloor \frac{mspan^p}{\scale} \rfloor+2n-4)\scale \leq mspan^p + (2n-4)\scale = mspan^p + (2n-4)\frac{\varepsilon \cdot \deadline}{2n} \leq mspan^p + \varepsilon\deadline \leq (1 + 2\varepsilon)mspan^p$ (recall: $\scale := \frac{\varepsilon \cdot \deadline}{2n}$).
	Therefore, a point $p'=(mspan^{p'},cost^{p'})$ with $mspan^{p'} \leq (1 + 2\varepsilon)mspan^p$ and $cost^{p'} \leq cost^{p}$ got reported.
	Setting $\varepsilon = 0.5 \alpha$ and noting that we repeat the process no more than $log(\sum_{j \in \jobs} \procSF{j})$ times concludes the proof. 
\end{proof}

\section{Future Work}
We give a small overview over the future research directions that emerge from our work.
$\mathbf{\modelnameextended}$:
If good approximations for $1\mid r_j\mid \sum w_j U_j$ become established, the algorithm given in \Cref{sec:extendedchain} for the extended chain could probably be improved.
One could model the incoming communication delay with release dates and get an equivalent subproblem to solve, instead of the approximate subproblem currently used.
$\mathbf{\modelname}$:
\Cref{sec:strong_hardness} gives a strong inapproximability result for the general case with regards to the cost function.
For two easy cases (chain and fully parallel graphs) we could establish FPTAS results, for graphs with a constant \dynConc we have an algorithm that finds optimal solutions with a $(1+\varepsilon)$ deadline augmentation.
Here one could explore if there are FPTAS results for different assumptions, are there approximation algorithms without resource augmentation for constant \dynConc instances, and lastly are there approximation algorithms with resource augmentation for the general case.
For the makespan function we already have a FPTAS for graphs with a constant \dynConc.
It remains to explore approximation algorithms or inapproximability results for the general case of this problem.
$\mathbf{\modelnameunit}$:
We show strong NP-hardness even for this simplified problem.
Since this is a special case of the general problem all constructive results still hold, additionally we were able to give a first simple algorithm for cost optimization in general graphs.
Here it would be interesting to look into more involved approximation algorithms that give better performance guarantees, maybe without resource augmentation.

\section*{Declarations}
\textbf{Conflict of interest:} The authors are not aware of any conflict of interests.









\bibliography{references}


\end{document}